\documentclass{aptpub}

\usepackage{graphicx}
\usepackage{amsfonts,amssymb,dsfont,color,bbm}
\usepackage{mathtools}
\usepackage{cite}
\usepackage{subcaption}
\usepackage{hyperref}
\usepackage{microtype}      
\usepackage[vlined,ruled,norelsize,lined,boxed]{algorithm2e}
\usepackage{tikz}
\usetikzlibrary{arrows.meta}

\allowdisplaybreaks

\DeclareMathAlphabet\mathbfcal{OMS}{cmsy}{b}{n}
\newcommand\hide[1]{}

\newcommand*\PR{\mathds{P}}
\newcommand*\OO{\mathcal{O}}
\newcommand*\EXP{\mathds{E}}

\newcommand{\bs}{\boldsymbol}

\authornames{Nima Akbarzadeh and Aditya Mahajan}
\shorttitle{Conditions for indexability of restless bandits and an algorithm to compute Whittle index}

\begin{document}

\title{Conditions for indexability of restless bandits and an $\OO(K^3)$ algorithm to compute Whittle index}

\authorone[McGill University]{Nima Akbarzadeh}%
\authorone[McGill University]{Aditya Mahajan}%
\addressone{Department of Electrical and Computer Engineering, McGill
University, 3480 Rue University, Montréal, QC H3A\,0E9.}
\footnote{This research was funded in part by the Innovation for Defence Excellence and Security (IDEaS) Program of the Canadian Department of National Defence through grant CFPMN2-037, and Fonds de Recherche du Quebec-Nature et technologies (FRQNT).}
\setcounter{footnote}{0}
\renewcommand\thefootnote{\arabic{footnote}}

\begin{abstract}
  Restless bandits are a class of sequential resource allocation problems
  concerned with allocating one or more resources among several alternative
  processes where the evolution of the process depends on the resource allocated to
  them. Such models capture the fundamental trade-offs between exploration and
  exploitation. In 1988, Whittle developed an index heuristic for restless
  bandit problems which has emerged as a popular solution approach due to its
  simplicity and strong empirical performance. The Whittle index heuristic is
  applicable if the model satisfies a technical condition known as
  indexability. In this paper, we present two general sufficient conditions
  for indexability and identify simpler to verify refinements of these
  conditions. We then revisit a previously proposed algorithm called adaptive greedy algorithm which is known to compute the Whittle index for a subclass of restless bandits. We show that a generalization of the adaptive greedy algorithm computes the Whittle index for all indexable restless bandits. \textcolor{black}{We present an efficient implementation of this algorithm which can compute the Whittle index of a restless bandit with $K$ states in $\OO(K^3)$ computations.} Finally, we present a detailed
  numerical study which affirms the strong performance of the Whittle index
  heuristic.
\end{abstract}



\section{Introduction}

Restless bandits are a class of sequential resource allocation problems
concerned with allocating one or more resources among several alternative
processes where the evolution of the process depends on the resource allocated
to them. Such models arise in various applications such as machine
maintenance~\cite{glazebrook2005index}, congestion
control~\cite{avrachenkov2013congestion}, healthcare~\cite{deo2013improving}, finance~\cite{glazebrook2013monotone}, channel
scheduling~\cite{liu2010indexability}, smart grid~\cite{abad2016near}, and
others. 

Restless bandits are a generalization of classical multi-armed
bandits~\cite{gittins2011multi}, where the processes remain frozen when
resources are not allocated to them. Gittins~\cite{gittins1979bandit} showed
that when a single resource is to be allocated among multiple processes, the
optimal policy has a simple structure: compute an index for each process and
allocate the resource to the process with the largest (or the lowest) index.
In contrast, the general restless bandit problem is
\textsc{pspace}-hard~\cite{papadimitriou1999complexity}.
Whittle~\cite{whittle1988restless} showed that index-based policies are
optimal for the Lagrangian relaxation of the restless bandit problem and
argued that the corresponding index, now called Whittle index, is a reasonable
heuristic for restless bandit problems. Subsequently, it has been shown that
the Whittle index heuristic is optimal under some
conditions~\cite{weber1990index,lott2000optimality} and performs well in
practice~\cite{ansell2003whittle,glazebrook2006some,glazebrook2002index}.

The Whittle index heuristic is applicable if a technical condition known as
\emph{indexability} is satisfied. 
Sufficient conditions for indexability have
been investigated under specific modeling assumptions: two state \hide{fully or partially observed }restless bandits~\cite{liu2010indexability,avrachenkov2013congestion}; monotone bandits~\cite{glazebrook2013monotone,
ansell2003whittle, avrachenkov2013congestion}; models with right-skip
free transitions~\cite{glazebrook2005index, glazebrook2006some}; models
with monotone or convex cost/reward~\cite{glazebrook2006some,
ansell2003whittle, avrachenkov2013congestion, archibald2009indexability,
yu2018deadline, ayesta2010modeling}; models satisfying partial
conservation laws~\cite{nino2001restless,nino2007dynamic}; and models arising in specific applications~\cite{glazebrook2002index, glazebrook2005index, glazebrook2006some, archibald2009indexability, ayesta2010modeling}. 

\textcolor{black}{Nino-Mora~\cite{ninomora2006, nino2007dynamic} proposed a generalization of Whittle index called marginal productivity index (MPI) for resource allocation problems where processes can be allocated fractional resources. In \cite{nino2007dynamic}, he also proposed an algorithm called the adaptive greedy algorithm, to compute the MPI when the model satisfies a technical condition called partial conservation laws (PCL). For restless bandits which satisfy the PCL condition, the Whittle index can be computed using the adaptive greedy algorithm. However, for general restless bandits, there are no known efficient algorithms to \textcolor{black}{exactly} compute the Whittle index. \textcolor{black}{It is possible to \emph{approximately}} compute the Whittle index by conducting a binary search over penalty for active action (or a subsidy for passive action)~\cite{qian2016restless,akbarzadeh2019dynamic} but such a binary search is computationally expensive because each step of the binary search requires solving a dynamic program.}

\textcolor{black}{In this paper, we revisit the restless bandit problem and
present three contributions. Our first contribution is to provide general
sufficient conditions for indexability which are based on an alternative characterization
of the passive set. We also present easy to verify refinements of these
sufficient conditions.}

\textcolor{black}{Our second contribution is to use a novel geometric interpretation of Whittle index to show that a refinement of the adaptive greedy algorithm proposed by Nino-Mora~\cite{nino2007dynamic} computes the Whittle index for all indexable restless bandits.} 
\textcolor{black}{We provide a computationally efficient implementation, which computes the Whittle indices of a restless bandit with $K$ states in $\OO(K^3)$ computations.}

\textcolor{black}{Our third contributions is to present three special cases:
(i)~\emph{Restless bandits with optimal threshold-based policy} which were previously studied in~\cite{ansell2003whittle, glazebrook2009index, avrachenkov2013congestion, glazebrook2013monotone, wang2019whittle, akbarzadeh2019restless}, (ii)~\emph{Stochastic monotone bandits} which may be considered as a generalization of monotone bandits~\cite{glazebrook2013monotone, ansell2003whittle, avrachenkov2013congestion}, and (iii)~\emph{Restless bandits with controlled restarts} similar to~\cite{akbarzadeh2019restless, wang2019whittle}, which is a generalizations of the restart models~\cite{glazebrook2005index, glazebrook2006some}. We show that these models are always indexable and the Whittle index can be computed in closed form.}


Finally, we present a detailed numerical study comparing the performance of the Whittle index policy with that of the optimal and myopic policies. Our study shows that in general, the performance of Whittle index policy is comparable to the optimal policy and considerably better than the myopic policy. 
\paragraph{Notation}

Uppercase letters ($X$, $Y$, etc.) denote random variables, 
lowercase letters ($x$, $y$, etc.) denote their realization, and 
script letters ($\mathcal{X}$, ${\cal Y}$, etc.) denote their
state spaces. Subscripts denote time: so, $X_t$ denotes a system variable at
time $t$ and $X_{1:t}$ is a short-hand for the system variables $(X_1, \dots,
X_t)$. $\PR(\cdot)$ denotes the probability of an event, $\EXP[\cdot]$ denotes
the expectation of a random variable. $\mathbb{Z}$ and $\mathbb{R}$ denote the sets of integers and real
numbers. Given a matrix~$P$, $P_{ij}$ denotes its $(i,j)$-th element. 



\section{Restless bandits: problem formulation and solution concept}
\label{sec:model}

\subsection{Restless Bandit Process}
A discrete-time restless bandit process (RB) is a controlled Markov
process~$(\mathcal{X}, \{0, 1\}, \allowbreak \{P(a)\}_{a \in \{0, 1\}}, c, x_0)$ where $\mathcal{X}$ denotes the state space which is a finite or countable set; $\{0, 1\}$ denotes the action space where the action~$0$ is called the \textit{passive} action and the action~$1$ is the \textit{active} action; $P(a)$, $a \in \{0, 1\}$, denotes the transition matrix when action~$a$ is chosen; $c: \mathcal{X} \times \{0, 1\} \to \mathbb{R}$ denotes the cost function; and $x_0$ denotes the initial state.
We use $X_t$ and $A_t$ to denote the action of the process at time~$t$. The
process evolves in a controlled Markov manner, i.e., for any realization
$x_{0:t+1}$ of $X_{0:t+1}$ and $a_{0:t+1}$ of $A_{0:t+1}$, we have 
$\PR( {X}_{t+1} = {x}_{t+1} | {X}_{0:t} = {x}_{0:t}, {A}_{0:t} = {a}_{0:t})
=
\PR ( X_{t+1} = x_{t+1} | X_{t} = x_{t}, A_{t} = a_t)$, which we denote by 
$P_{x_t x_{t+1}}(a_t)$.

\subsection{Restless Multi-armed Bandit Problem}
A restless multi-armed bandit is a collection of $n$ independent
RBs~$(\mathcal{X}^i, \{0, 1\}, \allowbreak \{P^i(a)\}_{a \in \{0, 1\}}, c^i, x^i_0)$,
$i \in {\cal N} \coloneqq \{1, \ldots, n\}$. A decision maker observes the state of
all RBs, may choose to activate only $m < n$ of them, and incurs a cost equal to the sum of the cost incurred by each RB.  

Let $\boldsymbol{\mathcal{X}} \coloneqq \prod_{i \in \cal N} \mathcal{X}^i$ and ${\mathbfcal A}(m) \coloneqq \bigl\{ {\bs a} = (a^1, \ldots, a^n) \in {\cal A}^{n} : \sum_{i \in {\cal N}} a^i = m\bigr\}$ denote the joint state space and the feasible action space, respectively. Let ${\boldsymbol X}_t \coloneqq (X^1_t, \dots X^n_t)$ and ${\boldsymbol A}_t = (A^1_t, \dots, A^n_t)$ denote the joint state and actions at time~$t$. As the RBs evolve independently, for any realization ${\boldsymbol x}_{0:t}$ of ${\boldsymbol X}_{0:t}$ and ${\boldsymbol a}_{0:t}$ of ${\boldsymbol A}_{0:t}$, we have 
$
\PR \left( {\bs X}_{t+1} = {\bs x}_{t+1} | {\bs X}_{0:t} = {\bs x}_{0:t}, {\bs A}_{0:t} = {\bs a}_{0:t} \right) = \prod_{i = 1}^{n} \PR \left( X^{i}_{t+1} = x^{i}_{t+1} | X^{i}_{t} = x^i_{t}, A^{i}_{t} = a^i_t \right).
$
When the system is in state ${\bs x}_t = (x^1_t, \ldots, x^n_t)$ and the decision-maker chooses action~${\bs a}_t = (a^1_t, \ldots, a^n_t)$, the system incurs a cost $\bar{c}({\bs x}_t, {\bs a}_t) \coloneqq \sum_{i \in {\cal N}} c^i(x^i_t, a^i_t)$.
The decision-maker chooses his actions using a time-homogeneous Markov policy
${\bs g}: {\mathbfcal X} \to {\mathbfcal A}(m)$, i.e., chooses ${\bs A}_t = {\bs g}({\bs X}_t)$.
The performance of any Markov policy ${\bs g}$ is given by 
\begin{equation*}
J^{({\bs g})}({\bs x}_0) \coloneqq (1-\beta) \EXP\biggl[ \sum_{t = 0}^{\infty} 
\beta^t \bar{c}({\bs X}_t, {\bs g}({\bs X}_t)) \bigg| {\bs X}_0 = {\bs x}_0
\biggr],
\end{equation*}
where $\beta \in (0, 1)$ is the discount factor and ${\bs x}_0$ is the initial state of the system.

We are interested in the following optimization problem.

\begin{problem}\label{prob:main}
  Given the discount factor $\beta \in (0,1)$, the total number $n$ of arms, the number $m$ of active arms, RBs~$(\mathcal{X}^i, \{0, 1\}, \{P^i(a)\}_{a \in \{0, 1\}}, c^i, x^i_0)$, $i \in {\cal N}$, and initial state~${\bs x}_0 \in {\mathbfcal X}$, choose a Markov policy ${\bs g} \colon {\mathbfcal X} \to {\mathbfcal A}(m)$ that minimizes $J^{({\bs g})}({\bs x}_0)$.
\end{problem}

Problem~\ref{prob:main} is a multi-stage stochastic control problem and one
can obtain an optimal solution using dynamic programming.
However, the dynamic programming solution is intractable for large~$n$ since
the cardinality of the state space is $\prod_{i \in {\cal N}}|\mathcal{X}^i|$,
which grows exponentially with~$n$. In the next section, we describe a
heuristic known as \textit{Whittle index} to efficiently obtain a suboptimal
solution of the problem.


\subsection{Indexability and the Whittle index}
Consider a RB~$(\mathcal{X}, \{0, 1\}, \{P(a)\}_{a \in \{0, 1\}}, c, x_0)$. For any $\lambda \in \mathbb{R}$, we consider a Markov decision process~$\{ \mathcal{X}, \{0, 1\}, \{P(a)\}_{a \in \{0, 1\}}, c_\lambda, x_0 \}$, where
\begin{equation} \label{eqn:modif_cost}
  c_\lambda(x, a) \coloneqq c(x, a) + \lambda a, 
  \quad \forall x \in \mathcal{X}, 
   \forall a \in \{0, 1\}.
\end{equation}
The parameter $\lambda$ may be viewed as a penalty for taking active action. The performance of any time-homogeneous policy $g: \mathcal{X} \to \{0, 1\}$ is 
\begin{equation}
J^{(g)}_{\lambda}(x_0) := (1-\beta) \EXP\biggl[ \sum_{t = 0}^{\infty} \beta^t c_\lambda(X_t, g(X_t)) \bigg| X_0  = x_0 \biggr]. \label{eqn:obj_func2}
\end{equation}

Consider the following optimization problem. 
\begin{problem}\label{prob:decompose}
	Given the RB~$(\mathcal{X}, \{0, 1\}, \{P(a)\}_{a \in \{0, 1\}}, c_\lambda, x_0)$ and the discount factor~$\beta \in (0,1)$, choose a Markov policy $g:
    \mathcal{X} \to \{0, 1\}$ to minimize $J_{\lambda}^{(g)}(x_0)$.
\end{problem}

Problem~\ref{prob:decompose} is also a Markov decision process and one can
obtain an optimal solution using dynamic programming. Let
$V_\lambda: \mathcal{X} \to \mathbb{R}$ be the unique fixed point of the
following: 
\begin{equation}
V_\lambda(x) =  \min \bigl\{ H_\lambda(x, 0) , H_\lambda(x, 1) \bigr\}, 
\quad\forall x \in \mathcal{X}, \label{eqn:bllmn_vf}
\end{equation}
where
\begin{equation}
  H_\lambda(x, a) =  (1-\beta) c_\lambda(x, a) + \beta \sum_{y \in \mathcal{X}} P_{xy}(a) V_\lambda(y), 
  \quad a \in \{0, 1\}. \label{eqn:def_H}
\end{equation}
Let $g_\lambda(x)$ denote the minimizer of the right hand side of $\eqref{eqn:bllmn_vf}$ where we set $g_\lambda(x) = 1$ if $H_\lambda(x, 0) = H_\lambda(x, 1)$. Then, from Markov decision theory~\cite{puterman2014markov}, we know that the time-homogeneous policy $g_\lambda$ is optimal for Problem~\ref{prob:decompose}.

Define the \emph{passive set} ${\Pi}_\lambda$ to be the set of states where passive action is optimal, i.e., 
\begin{equation} \label{eqn:pass_set}
{\Pi}_\lambda \coloneqq \left\{ x \in \mathcal{X}: g_\lambda(x) = 0 \right\}.
\end{equation}

\begin{definition}[Indexability] \label{def:indexability}
	An RB is indexable if ${\Pi}_\lambda$ is increasing in $\lambda$, i.e., for any $\lambda', \lambda'' \in \mathbb{R}$,
    \(\lambda' \leq \lambda''\) implies that \({\Pi}_{\lambda'} \subseteq
    {\Pi}_{\lambda''}\).
\end{definition}
\begin{definition}[Whittle index] \label{def:index}
	The Whittle index of state~$x$ of an indexable RB is the smallest value of $\lambda$ for which $x$ is part of the passive set ${\Pi}_{\lambda}$, i.e.,
    \(
	w(x) = \inf \left\{ \lambda \in \mathbb{R} : x \in {\Pi}_{\lambda} \right\}.
  \)
\end{definition}
Alternatively, the Whittle index $w(x)$ is a value of the penalty~$\lambda$ for which the optimal policy is indifferent between taking active and passive action when the RB is in state~$x$. 

\subsection{Whittle Index Heuristic}

A restless multi-armed bandit problem is said to be indexable if all RBs are
indexable. For indexable problems, the Whittle index heuristic is as follows:
\emph{Compute the Whittle indices of all arms offline. Then, at each time,
  obtain the Whittle indices of the current state of all arms and play
arms with the $m$ largest Whittle indices}. 

As mentioned earlier, Whittle index policy is a popular approach for restless
bandits because: (i)~its complexity is linear in the number of alternatives
and (ii)~it often performs close to optimal in
practice~\cite{ansell2003whittle,glazebrook2006some,glazebrook2002index}.
However, there are only a few general conditions to check indexability for
general models.

\subsection{Alternative characterizations of passive set}

We now present alternative characterizations of passive set, which is
important for the sufficient conditions of indexability that we provide later.

\textcolor{black}{Let $\Sigma$ denote the family of all stopping times with respect to the
natural filtration of $\{X_t\}_{t \ge 0}$. For any
state~$x \in \mathcal{X}$, penalty $\lambda \in \mathbb{R}$, and stopping time~$\tau \in \Sigma$, define
\begin{align}
  M(x, \tau) &\coloneqq \EXP\bigl[ \beta^{\tau} | X_0 = x,
    \{ A_t = 0 \}_{t = 0}^{\tau-1} \bigr],
  \notag \\
  L(x, \tau) &\coloneqq \EXP\Big[ \sum_{t = 0}^{\tau-1} \beta^{t}
    c(X_t, 0) + \beta^{\tau} c(X_{\tau}, 1) \Bigm| X_0 = x, 
    \{ A_t = 0 \}_{t = 0}^{\tau-1} \Big],
  \notag \\
  W_\lambda(x) &\coloneqq (1-\beta)\lambda 
  + \beta \sum_{y \in \mathcal{X}} P_{xy}(1) V_\lambda(x).
  \label{eqn:W_def}
\end{align}
Let $h_{\tau, \lambda}$
denote the (history dependent) policy that takes passive action up to time
$\tau-1$, active action at time~$\tau$, and then follows the optimal
policy~$g_\lambda$ (for Problem~\ref{prob:decompose}).} 
We now present different characterizations of the passive set.


\begin{proposition} \label{prop:characterization}
	The following characterizations of the passive set are equivalent.
    \begin{itemize}
      \item $\Pi^{(a)}_\lambda = \{x \in \mathcal{X} : g_\lambda(x) = 0 \}$
      \item $\Pi^{(b)}_\lambda = \{x \in \mathcal{X} : H_\lambda(x, 0) < H_\lambda(x, 1) \}$
      \item $\Pi^{(c)}_\lambda = \{x \in \mathcal{X} : \exists \sigma \in
        \Sigma, \sigma \neq 0, \text{such that } J^{(h_{\sigma, \lambda})}_{\lambda}(x) < J^{(h_0)}_{\lambda}(x) \}$
      \item $\Pi^{(d)}_\lambda = \{x \in \mathcal{X} : \exists \sigma \in
        \Sigma, \sigma \neq 0, \text{such that } (1-\beta) \left( L(x, \sigma) - c(x, 1) \right) < W_\lambda(x) - \EXP [ \beta^{\sigma} W_\lambda(X_\sigma) | X_0 = x ] \}$
    \end{itemize}
\end{proposition}
See Appendix~\ref{app:characterization} for proof.


\section{Sufficient Conditions for Indexability} \label{sec:indexability}
In this section, we identify sufficient conditions for a RB to be indexable.

\subsection{Preliminary results}

Consider a RB~$(\mathcal{X}, \{0, 1\}, \{P(a)\}_{a \in \{0, 1\}}, c, x_0)$. For
any Markov policy $g \colon \mathcal{X} \to \{0,1\}$ and $\lambda \in \mathbb{R}$, we
can write
\begin{equation} \label{eq:J-split}
  J^{(g)}_\lambda(x) = D^{(g)}(x) + \lambda N^{(g)}(x),
\end{equation}
where
\begin{align*}
  D^{(g)}(x) &\coloneqq (1-\beta) \EXP \bigg[ \sum_{t = 0}^{\infty} \beta^t
  c(X_t, g(X_t)) \bigg| X_0 = x \bigg]
  \\
  \shortintertext{and}
  N^{(g)}(x) &\coloneqq (1-\beta) \EXP \bigg[ \sum_{t = 0}^{\infty} \beta^t
  g(X_t) \bigg| X_0 = x \bigg] \label{eqn:N_def}
\end{align*}
are the expected discounted total cost and the expected number of activations
under policy~$g$ starting at initial state~$x$. 
$D^{(g)}(\cdot)$ and $N^{(g)}(\cdot)$ can be computed using policy
evaluation formulas. In particular, define $P^{(g)} \colon \mathcal{X} \times
\mathcal{X} \to \mathbb{R}$ and $c^{(g)} \colon \mathcal{X} \to \mathbb{R}$ as
follows: $P^{(g)}_{xy} = P_{xy}(g(x)) \text{  and  } c^{(g)}_\lambda(x) =
c_\lambda(x, g(x)) = c^{(g)}(x, g(x)) + \lambda g(x)$ for any $x \in
\mathcal{X}$. We also view $g$ as an element in $\{0,1\}^{|\mathcal{X}|}$. 
Then, using the policy evaluation formula for infinite horizon
MDPs~\cite{puterman2014markov}, we obtain
\begin{equation} \label{eqn:DN_LinComp}
  D^{(g)}(x) = (1-\beta)\bigl[(I - \beta P^{(g)})^{-1} c^{(g)}\bigr](x)
  \text{ and }
  N^{(g)}(x) =(1-\beta) \bigl[(I - \beta P^{(g)})^{-1} g\bigr](x).
\end{equation}

\begin{figure}[!t] 
	\centering
	\begin{tikzpicture}[x=1.75cm, y=1.5cm, font=\sffamily\small,>={Latex[width=1.5mm, length=1.5mm]}] 
	\draw[->] (-0.5,0) -- (2.5,0) node[below] (xaxis) {$\lambda$};
	\draw[->] (0,-0.25) -- (0,2) node[left] (yaxis) {$J^{(\cdot)}_\lambda(x)$};
	\draw[-, line width=0.05mm] (0.43, 0.61) -- (1.5,1.9);
	\draw[-, line width=0.05mm] (-0.43, 0.19) -- (0.43, 0.61);
	\draw[-, line width=0.05mm] (1.25, 1.03) -- (2,1.4);
	\draw[-, line width=0.05mm] (1.25, 1.03) -- (2,1.1);
	\draw[-, line width=0.05mm] (-1.25, 0.77) -- (1.25, 1.03);
	\draw[-, line width=0.4mm] (-0.43,-0.41) -- (0.43, 0.61);
	\draw[-, line width=0.4mm] (0.43, 0.6) -- (1.25, 1.03);
	\draw[-, line width=0.4mm] (1.25, 1.03) -- (2,1.1);
	\draw[dashed] (0.43, 0.61) -- (0.43, 0);
	\draw[dashed] (1.25, 1.03) -- (1.25, 0);
	\draw[fill=white] (0.43,0) circle [radius = 0.02cm] node[below] {$\lambda^{12}_{c}$};
	\draw[fill=white] (1.25,0) circle [radius = 0.02cm] node[below] {$\lambda^{23}_{c}$};
	\draw[fill] (0,0.1) circle [radius = 0.02cm] node[left] {$D^{(h_2)}(x)$};
	\draw[fill] (0,0.4) circle [radius = 0.02cm] node[left] {};
	\draw[fill] (0,0.43) circle [radius = 0.00001mm] node[left] {$D^{(h_2)}(x)$};
	\draw[fill] (0,0.9) circle [radius = 0.02cm] node[left] {};
	\draw[fill] (0,0.99) circle [radius = 0.00001mm] node[left] {$D^{(h_3)}(x)$};
	\draw (1.5,1.9) circle [radius = 0.00001mm] node[right] {$J^{(h_1)}_\lambda(x)$};
	\draw (2,1.4) circle [radius = 0.00001mm] node[right] {$J^{(h_2)}_\lambda(x)$};
	\draw (2,1.1) circle [radius = 0.00001mm] node[right] {$J^{(h_3)}_\lambda(x)$};
	\end{tikzpicture}
	\caption{\textcolor{black}{An illustration of the plot of $J^{(\cdot)}_\lambda(x)$ versus $\lambda$ for $g \in {\cal G} := \{h_1, h_2, h_3\}$. Let $\lambda^{ij}_c$ denote the $\lambda$-value of the intersection of $J^{(h_i)}_\lambda(x)$ and $J^{(h_j)}_\lambda(x)$. Note that in this plot, for all $\lambda \in ( -\infty, \lambda^{12}_c ]$ the policy $h_1$ is optimal; for all $\lambda \in [ \lambda^{12}_c, \lambda^{23}_c ]$ the policy $h_2$ is optimal; and for all $\lambda \in [ \lambda^{23}_c, \infty)$ the policy $h_3$ is optimal. The lower concave envelope of $J^{(h_i)}_\lambda(x)$ (shown as a thick line) is the value function $V_\lambda(x)$, which is piecewise linear, concave, increasing and continuous.}}
	\label{fig:DN}
\end{figure}
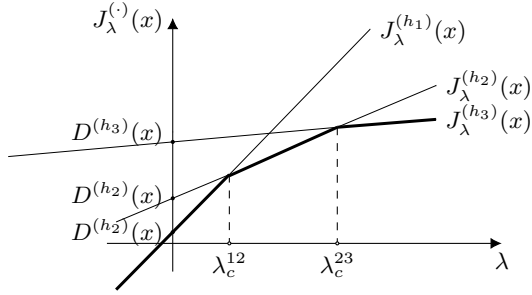

\textcolor{black}{
We now provide a geometric interpretation of the value function $V_\lambda(x)$ as a function of~$\lambda$. For any $g \in {\cal G}$, the plot of $J^{(g)}_\lambda(x) = D^{(g)}(x) + \lambda N^{(g)}(x)$ as a function of $\lambda$ is a straight line with $y$-intercept $D^{(g)}(x)$ and slope $N^{(g)}(x)$. By definition, 
\( V_\lambda(x) = \inf_{g \in {\cal G}} J^{(g)}_\lambda(x). \)
Thus, $V_\lambda(x)$ is the lower concave envelope of the family of straight lines $\{J^{(g)}_\lambda(x)\}_{g \in {\cal G}}$. See Fig.~\ref{fig:DN} for an illustration. Thus, we have the following:
\begin{lemma}\label{lem:V-cts}
  For any $x \in \mathcal{X}$, $V_\lambda(x)$ is continuous, increasing, \textcolor{black}{piece-wise linear} and concave in $\lambda$. Furthermore, when ${\cal X}$ is finite, $V_\lambda(x)$ is piecewise linear. 
\end{lemma}
\begin{proof}
  For any Markov policy $g$, $N^{(g)}(x)$ is non-negative. Therefore, $J^{(g)}_\lambda(x) = D^{(g)}(x) + \lambda N^{(g)}(x)$ is increasing and continuous in $\lambda$. Since $V_\lambda(x)$ is an infimization of a family of linear functions, it is concave (see Fig.~\ref{fig:DN}). In addition, as monotonicity and continuity are preserved under infimization, the value function is also increasing and continuous in $\lambda$.
  Finally, when ${\cal X}$ is finite, there are only finite number of pieces. Thus, $V_\lambda(x)$ is the minimum of a finitely many linear functions and hence, piece-wise linear.
\end{proof}
}
\begin{lemma}\label{lem:V-diff}
  For any $\lambda', \lambda'' \in \mathbb{R}$, 
  \[
    (\lambda'' - \lambda') N^{(g_{\lambda''})}(x) 
    \le 
    V_{\lambda''}(x) - V_{\lambda'}(x) 
    \le
    (\lambda'' - \lambda') N^{(g_{\lambda'})}(x),
    \quad
    \forall x \in \mathcal{X}.
  \]
  Consequently, $N^{(g_\lambda)}(x)$ is non-increasing in~$\lambda$. 
\end{lemma}
\begin{proof}
  Recall that $V_\lambda(x) = J^{(g_\lambda)}_\lambda(x) \le
  J^{(g_{\lambda'})}_\lambda(x)$ for any $\lambda' \neq \lambda$. 
  Thus,
  \begin{align}
    V_{\lambda''}(x) - V_{\lambda'}(x)
    &=   J^{(g_{\lambda''})}_{\lambda''}(x) - J^{(g_{\lambda'})}_{\lambda'}(x) 
    \le  J^{(g_{\lambda'})}_{\lambda''}(x) - J^{(g_{\lambda'})}_{\lambda'}(x) 
    \notag \\
    &\stackrel{(a)}=
    (\lambda'' - \lambda')N^{(g_{\lambda'})}(x),
  \end{align}
  where $(a)$ follows from~\eqref{eq:J-split}. Similarly, we have 
  \begin{align}
    V_{\lambda''}(x) - V_{\lambda'}(x)
    &=   J^{(g_{\lambda''})}_{\lambda''}(x) - J^{(g_{\lambda'})}_{\lambda'}(x) 
    \ge  J^{(g_{\lambda''})}_{\lambda''}(x) - J^{(g_{\lambda''})}_{\lambda'}(x) 
    \notag \\
    &\stackrel{(a)}=
    (\lambda'' - \lambda')N^{(g_{\lambda''})}(x),
  \end{align}
  where $(a)$ follows from~\eqref{eq:J-split}. The result follows from
  combining the above inequalities.
\end{proof}

\subsection{Sufficient conditions for indexability}

\begin{theorem} \label{Thm:suf_1}
  Define $\mathcal H = \{ (g,h) : g, h \colon \mathcal X \to \{0, 1\} 
  \text{ such that for all } x \in \mathcal X, N^{(g)}(x) \ge N^{(h)}(x) \}$. 
	Each of the following is a sufficient condition for Whittle indexability:
	\begin{itemize}
      \item[a.] For any $g,h \in \mathcal H$, we have that for every $x, z
        \in \mathcal{X}$,
        \begin{equation}
          \sum_{y \in \mathcal{X}}
          \Bigl\{
            \bigl[ \beta P_{zy}(1) - P_{xy}(1) \bigr]^+ N^{(g)}(y) 
            -
            \bigl[ P_{xy}(1) - \beta P_{zy}(1) \bigr]^+ N^{(h)}(y) 
          \Bigr\}
          \le \frac{(1-\beta)^2}{\beta}.
          \label{eqn:P1_cond}
        \end{equation}
      \item[b.] For any $g, h \in \mathcal H$, we have that for every $x \in
        \mathcal X$, 
        \begin{equation}
          \sum_{y \in \mathcal{X}} 
          \Bigl\{
            \bigl[ P_{xy}(0) - P_{xy}(1) \bigr]^+ N^{(g)}(y) 
            -
            \bigl[ P_{xy}(1) - P_{xy}(0) \bigr]^+ N^{(h)}(y) 
          \Bigr\}
          \le \frac{1-\beta}{\beta}.
        \label{eqn:P2_cond}
        \end{equation}
	\end{itemize}
\end{theorem}
See Appendix~\ref{prf:Thm:suf_1} for the proof. 
The sufficient conditions of Theorem~\ref{Thm:suf_1} can be difficult to verify.
Simpler sufficient conditions are stated below.

\begin{proposition} \label{prop:1}
	Each of the following is a sufficient condition for~\eqref{eqn:P1_cond}.
	\begin{enumerate}
		\item[a.] $\max_{x, z \in \mathcal{X}} \sum_{y \in \mathcal{X}} \bigl[
          \beta P_{z y}(1) - P_{x y}(1) \bigr]^{+} \leq (1 - \beta)^2/\beta$.
		\item[b.] $P_{xy}(1) = P_{zy}(1)$, for any $x, y, z \in \mathcal{X}$.
	\end{enumerate}
	In addition, each of the following is a sufficient condition for~\eqref{eqn:P2_cond}.
	\begin{enumerate}
		\item[c.] $\max_{x \in \mathcal{X}} \sum_{y \in \mathcal{X}} \left[ P_{xy}(0) - P_{xy}(1) \right]^{+} \leq (1 - \beta)/\beta$.
        \item[d.] $\beta \le 0.5$.
	\end{enumerate}
\end{proposition}
See Appendix~\ref{prf:prop:1} for proof.

\paragraph*{Some remarks}

\begin{enumerate}
 	\item The sufficient conditions of Theorem~\ref{Thm:suf_1} and
    Proposition~\ref{prop:1} a, c, d  may be viewed as bounds on the discount
    factor $\beta$ for which a RB is indexable. Numerical experiments to explore such a property are presented in \cite{nino2007dynamic}. \textcolor{black}{A qualitatively similar result was established in~\cite[Corollary 5]{nino2001restless} which showed that a restless bandit process is GCL (Generalized Conservation Laws) indexable for sufficiently small discount factors.} GCL indexability is a sub-class of PCL indexability, which is a sub-class of Whittle indexability. Thus, Proposition \ref{prop:1} provides a quantitative characterization of the qualitative observation made in \cite{bertsimas1996conservation} and generalizes it to a broader class of models.

 	\item We refer to models that satisfy the sufficient condition of
    Proposition~\ref{prop:1}.b as \textit{restless bandits with controlled
    restarts}. Such models arise in various scheduling problems (e.g., machine
    maintenance, surveillance, etc.) where taking the active action resets the
    state according to known probability distribution. Specific instances of
    such models are considered in~\cite{akbarzadeh2019restless,
    wang2019whittle}. The special case when the active action resets to a
    specific (pristine) state are considered in~\cite{glazebrook2005index,
    glazebrook2006some}.

	\item \textcolor{black}{A different class of restart models have been considered in \cite{Jacko2012,ayesta2010modeling,avrachenkov2013congestion} where the passive action resets the state of the arm. Note that such models do not satisfy Proposition~\ref{prop:1}b and additional modeling assumptions are required to establish indexability. See \cite{Jacko2012,ayesta2010modeling,avrachenkov2013congestion} for details.}
\end{enumerate}

\section{An algorithm to compute Whittle index} \label{sec:computation}

Given an indexable RB, a naive method to compute Whittle index at state~$x$
is to do a binary search over the penalty~$\lambda$ and find the critical
penalty $w(x)$ such that for $\lambda \in (-\infty, w(x))$, $g_\lambda(x) = 0$
and for $\lambda \in [w(x), \infty)$, $g_\lambda(x) = 1$. Although such an
approach has been used in the
literature~\cite{qian2016restless,akbarzadeh2019dynamic}, it is not efficient
as it requires a separate binary search for each state. 
\textcolor{black}{For a sub-class of restless bandits which satisfy an additional technical condition called partial conservation law (PCL), Nino-Mora~\cite{nino2002dynamic,nino2007dynamic} presented an algorithm called adaptive greedy algorithm to compute the Whittle index. 
In this section, we present an algorithm that may be viewed as a refinement of the adaptive greedy algorithm and show that it computes the Whittle index for all indexable RBs. The result of this section are restricted to the case of finite ${\cal X}$.}

Let $K$ denote the number of states (i.e., $K = |{\cal X}|$) and $K_D (\leq K)$ denote the number of distinct Whittle indices. Let $\Lambda^* = \{\lambda_1, \dots, \lambda_{K_D}\}$ where $\lambda_1 < \lambda_2 < \dots < \lambda_{K_D}$ denote the sorted list of distinct Whittle indices. Also, let $\lambda_0 = -\infty$. For any $d \in \{0, \ldots, K_D\}$, let ${\cal P}_d := \{x \in {\cal X}: w(x) \leq \lambda_d\}$ denote the set of states with Whittle index less than or equal to $\lambda_d$. Note that ${\cal P}_0 = \emptyset$ and ${\cal P}_{K_D} = {\cal X}$. \textcolor{black}{Let $\Gamma_{d+1} = {\cal P}_{d+1}\backslash{\cal P}_d$ denote the states with Whittle index $\lambda_{d+1}$.}

\textcolor{black}{For any subset ${\cal S} \subseteq {\cal X}$, define the policy $\bar{g}^{(\cal S)}: {\cal X} \to \{0, 1\}$ as 
\begin{equation} \label{eqn:gbar}
	\bar{g}^{(\cal S)}(x) = 
	\begin{cases}
	0, ~ \text{ if } x \in {\cal S} \\
	1, ~ \text{ if } x \in {\cal X}\backslash{\cal S}.
	\end{cases}
\end{equation}
Thus, the policy $\bar{g}^{(\cal S)}$ takes passive action in set ${\cal S}$ and active action in set ${\cal X}\backslash{\cal S}$.}

Now for any $d \in \{0, \ldots, K_D-1\}$, and all states $y \in {\cal X}\backslash{\cal P}_{d}$, define $h_d = \bar{g}^{({\cal P}_d)}$, $h_{d,y} = \bar{g}^{({\cal P}_d \cup \{y\})}$ and for all $x \in \Lambda_{d, y}$,
\begin{align} \label{eq:mu}
\Lambda_{d, y} = \{x \in {\cal X}: N^{(h_d)}(x) \neq N^{({h}_{d,y})}(x) \}, ~ \mu_{d,y}(x) = \dfrac{D^{({h}_{d, y})}(x) - D^{(h_d)}(x)}{N^{(h_d)}(x) - N^{({h}_{d, y})}(x)}. 
\end{align}

\textcolor{black}{
\begin{lemma} \label{lemma:WJ}
	For \textcolor{black}{an indexable RB with} $d \in \{0, \ldots, K_D-1\}$, we have the following:
	\begin{enumerate}
		\item For all \textcolor{black}{$y \in \Gamma_{d+1}$}, we have $w(y) = \lambda_{d+1}$.
		\item For all $y \in {\cal X}\backslash{\cal P}_d$ and $\lambda \in (\lambda_{d}, \lambda_{d+1}]$, we have $J^{(h_{d,y})}_\lambda(x) \geq J^{(h_d)}_\lambda(x)$ for all $x \in {\cal X}$ with equality if and only if \textcolor{black}{$y \in \Gamma_{d+1}$} and $\lambda = \lambda_{d+1}$.
	\end{enumerate}
\end{lemma}
\begin{proof}
	See Appendix~\ref{prf:lemma:WJ}.
\end{proof}
}

\textcolor{black}{
\begin{theorem}\label{Thm:Widx}
	\textcolor{black}{For an indexable RB,} the following properties hold:
	\begin{enumerate}
      \item For any \textcolor{black}{$y \in \Gamma_{d+1}$},
        the set $\Lambda_{d, y}$ is non-empty.
		\item For any $x \in \Lambda_{d, y}$, $\mu_{d,y}(x) \geq \lambda_{d+1}$ with equality if and only if \textcolor{black}{$y \in \Gamma_{d+1}$}.
	\end{enumerate}
\end{theorem}
\begin{proof}
	The proof of each part is as follows:
	\begin{enumerate}
      	\item We prove the result by contradiction. Suppose that there exists a \textcolor{black}{$y \in \Gamma_{d+1}$}, such that $\Lambda_{d, y} = \emptyset$ which means $N^{(h_d)}(x) = N^{({h}_{d, y})}(x)$ for all $x \in {\cal X}$. By Lemma~\ref{lemma:WJ}, we have that $J^{(h_d)}_{\lambda_{d+1}}(x) = J^{({h}_{d, y})}_{\lambda_{d+1}}(x)$. Therefore, from~\eqref{eq:J-split} we infer $D^{(h_d)}(x) = D^{({h}_{d, y})}(x)$ for all $x \in {\cal X}$. Since both $D^{(g)}(x)$ and $N^{(g)}(x)$ do not depend on $\lambda$, \eqref{eq:J-split} implies that for any $\lambda$ and $x \in \mathcal{X}$, we have $J^{(h_d)}_{\lambda}(x) = J^{({h}_{d, y})}_{\lambda}(x)$. This implies that the policies $h_d$ and $h_{d,y}$ will be optimal for the same set of~$\lambda$. Now, since policy $h_d$ is optimal for all $\lambda \in (\lambda_d, \lambda_{d+1}]$ (by definition), so is $h_{d,y}$. Hence $y \in \mathcal{P}_d$. But we started by assuming that $y \not\in \mathcal{P}_d$, so we have a contradiction.
       	\item By Lemma~\ref{lemma:WJ}, part~2, for all $y \in {\cal X}\backslash{\cal P}_d$, $\lambda \in (\lambda_{d}, \lambda_{d+1}]$ and for all $x \in \Lambda_{d, y}$, we have $J^{(h_{d,y})}_\lambda(x) \geq J^{(h_d)}_\lambda(x)$. Then, by \eqref{eq:J-split} we infer 
       	\[D^{(h_{d,y})}(x) + \lambda N^{(h_{d,y})}(x) \geq D^{(h_d)}(x) + \lambda N^{(h_d)}(x).\]
       	Finally, we have $\mu_{d,y}(x) \geq \lambda$ and thus, $\mu_{d,y}(x) \geq \lambda_{d+1}$ for all $x \in \Lambda_{d, y}$. 
		This proves the first part of the statement. To prove the second part, note that policy~$h_d$ is an optimal policy for $\lambda \in (\lambda_d, \lambda_{d+1}]$ and for any $y \in \mathcal{P}_{d+1}$, the policy~${h}_{d, y}$ is an optimal policy for $\lambda \in (\lambda_{d+1}, \lambda_{d+2}]$.
		From Lemma~\ref{lem:V-cts}, we know that $V_\lambda(x)$ is continuous in $\lambda$ for all $x \in \mathcal{X}$. Thus, for all $x \in \mathcal{X}$,
		\[
		\lim_{\lambda \uparrow \lambda_{d+1}} J^{(h_d)}_\lambda(x) = \lim_{\lambda \uparrow \lambda_{d+1}} V_\lambda(x) 
		= \lim_{\lambda \downarrow \lambda_{d+1}} V_\lambda(x) = \lim_{\lambda \downarrow \lambda_{d+1}} J^{(h_{d, y})}_\lambda(x).
		\]
 		Thus, for all $x \in \mathcal{X}$, $J^{(h_d)}_{\lambda_{d+1}}(x) = J^{({h}_{d, y})}_{\lambda_{d+1}}(x)$ and, therefore, 
 		\begin{equation*}
 		D^{(h_d)}(x) + \lambda_{d+1} N^{(h_d)}(x) 
 		= 
 		D^{({h}_{d, y})}(x) + \lambda_{d+1} N^{(h_{d, y})}(x).
 		\end{equation*}
 		As a result, $\lambda_{d+1} = \mu_{d,y}(x)$ for all $x \in \Lambda_{d, y}$. 
 	\end{enumerate}
\end{proof}
}

Theorem \ref{Thm:Widx} suggests a method for identifying the Whittle index of \textcolor{black}{any indexable RB} by iteratively identifying the set ${\cal P}_d$ and the Whittle index $\lambda_d$. By definition, ${\cal P}_0 = \emptyset$ and $\lambda_0 = -\infty$. Now suppose ${\cal P}_0 \subset {\cal P}_1 \subset \ldots \subset {\cal P}_{d}$ and $\lambda_0 < \lambda_1 < \ldots < \lambda_d$ have been identified. We will describe how to determine ${\cal P}_{d+1}$ and $\lambda_{d+1}$.
\begin{enumerate}
	\item For $h_d = {\bar g}^{({\cal P}_d)}$, compute $N^{(h_d)}$ by solving \eqref{eqn:DN_LinComp}.
	\item For all $y \in {\cal X}\backslash{\cal P}_d$, compute $N^{(h_{d, y})}$ where $h_{d, y} = \bar{g}^{({\cal P} \cup \{y\})}$ by solving \eqref{eqn:DN_LinComp} and compute $\Lambda_{d, y}$. Let $\mu^*_{d, y} = \min_{x \in \Lambda_{d, y}} \mu_{d, y}(x)$ where $\mu_{d,y}(x)$ is given by Theorem~\ref{Thm:Widx}. Then, $\lambda_{d+1} = \min_{y \in {\cal X}\backslash{\cal P}_d} \mu^*_{d, y}$, \textcolor{black}{$\Gamma_{d+1} = \arg \min_{y \in {\cal X}\backslash{\cal P}_d} \mu^*_{d, y}$, and we get ${\cal P}_{d+1} = {\cal P}_d \cup \Gamma_{d+1}$} (recall that argmin denotes the set of all minimizers) and $w(x) = \lambda_{d+1}$, \textcolor{black}{$\forall x \in \Gamma_{d+1}$}.
\end{enumerate}
Iteratively proceeding this way, we can compute the Whittle index for all states. The detailed algorithm is presented in Algorithm~\ref{Alg:Widx_comp}.

\begin{algorithm}[!t]
  \DontPrintSemicolon
  \SetKwInOut{Input}{input}
  \Input{RB~$(\mathcal{X}, \{0, 1\}, {P(a)}_{a \in \{0, 1\}}, c, x_0)$,
  discount factor $\beta$.} 
  Initialize $d = 0$ and $\mathcal P_0 = \emptyset$. \;
  \While{$\mathcal{P}_d \neq \mathcal{X}$}{
  	Compute $\Lambda_{d, y}$ and $\mu_{d,y}(x)$ using~\eqref{eq:mu}, $\forall y \in \mathcal{X}\setminus \mathcal{P}_d$. \;
    Compute $\mu^*_{d,y} = \min_{x \in \Lambda_{d,y}} \mu_{d,y}(x)$, $\forall y \in \mathcal{X}\setminus \mathcal{P}_d$. \;
    Compute $\lambda_{d+1} = \min_{y \in {\cal X}\backslash {\cal P}_d} \mu^*_{d,y}$. \;
    Compute $\Gamma_{d+1} = \arg \min_{y \in {\cal X}\backslash {\cal P}_d} \mu^*_{d,y}$. \;
    Set $w(z) = \lambda_{d+1}$, $\forall z \in \Gamma_{d+1}$. \;
    Set ${\cal P}_{d+1} = {\cal P}_d \cup \Gamma_{d+1}$. \;
	Set $d = d+1$. \;
  }
  \caption{Computing Whittle index of all states of an indexable RB}
  \label{Alg:Widx_comp}
\end{algorithm}


\subsection{An efficient implementation using Sherman-Morrison formula} \label{subsec:eff}

\textcolor{black}{We now present an efficient implementation of Algorithm~\ref{Alg:Widx_comp} using the Sherman-Morrison inverse formula. Suppose $A \in \mathbb {R} ^{n\times n}$ is an invertible square matrix, $u,v\in \mathbb {R} ^{n}$ are column vectors, such that $A + u v^{\textsf{T}}$ is invertible. Then, the Sherman-Morrison inverse formula is 
  \begin{equation}\label{eq:SM}
    \bigl(A+uv^{\textsf {T}}\bigr)^{-1}=A^{-1}-{\frac{A^{-1}uv^{\textsf {T}}A^{-1}} {1+v^{\textsf {T}}A^{-1}u}}. 
\end{equation}
  Furthermore, given $b \in \mathbb{R}^n$, if $x$ is the solution of $Ax = b$ and $y$ is the solution $Ay = u$, then the solution of $(A + u v^{\textsf{T}}) \tilde x = b$ is given by (see~\cite[Corollary 2]{egidi2006sherman})
\begin{equation}\label{eq:SM-soln}
  \tilde x = x - \frac{v^{\mathsf{T}} x}{1 + v^{\mathsf{T}} y} y.
\end{equation}}

\textcolor{black}{Note that for any Markov policy~$g$, $(I - \beta P^{(g)})$ is invertible because $\beta P^{(g)}$ is a sub-stochastic matrix and has a spectral radius less than~$1$. Therefore, the conditions of using the Sherman-Morrison formula are satisfied. Hence, using~\eqref{eq:SM-soln}, Eq.~\eqref{eqn:DN_LinComp} may be written (in matrix form) as
\begin{equation} \label{eqn:DN_PhiComp}
	D^{(g)} = (1 - \beta) \Phi^{(g)} c^{(g)} \text{ and } N^{(g)} = (1 - \beta) \Phi^{(g)} g.
\end{equation}}

\textcolor{black}{Now, for any $d \in \{0, \ldots, K_D-1\}$ and a state~$y \in {\cal X} \backslash {\cal P}_d$, consider policies~$h_d = \bar{g}^{({\cal P}_d)}$ and $h_{d,y} = \bar{g}^{({\cal P}_{d,y})}$.}
\textcolor{black}{Let $e_y$ denote the unit vector with $1$ in the $y$-th location and $\rho_y$ be a vector given by
  \( [\rho_y]_x = P_{yx}(1) - P_{yx}(0) \), for all $x \in {\cal X}$.
Then, $P^{(h_{d,y})} = P^{(h_d)} - e_y \rho^{\mathsf{T}}_y$. Therefore, 
\[ I - \beta P^{(h_{d,y})} = \bigl( I - \beta P^{(h_{d})} \bigr) + \beta e_y \rho^{\mathsf{T}}_y. \]
Let $\Phi^{(h_{d})}_{\cdot y}$ denote the $y$-th column of $\Phi^{(h_{d})}$, i.e., $\Phi^{(h_{d})}_{\cdot y} = \Phi^{(h_{d})} e_y$. Then, by Sherman-Morrison inverse formula~\eqref{eq:SM} and~\eqref{eq:SM-soln}, we have
\begin{gather}
	\Phi^{(h_{d,y})} = \Phi^{(h_{d})} -  \dfrac{\beta\Phi^{(h_{d})} e_y \rho^{\mathsf{T}}_y \Phi^{(h_{d})} }{1 + \beta \rho^{\mathsf{T}}_y \Phi^{(h_{d})}_{\cdot y}} 
    \label{eqn:phi_update}
    \\ 
	D^{(h_{d,y})}  = D^{(h_{d})} - \dfrac{\beta\rho^{\mathsf{T}}_y D^{(h_{d})} }{1 + \beta\rho^{\mathsf{T}}_y \Phi^{(h_{d})}_{\cdot y}} \Phi^{(h_{d})}_{\cdot y},  \quad
	N^{(h_{d,y})}  = N^{(h_{d})} - \dfrac{\beta\rho^{\mathsf{T}}_y N^{(h_{d})} }{1 + \beta\rho^{\mathsf{T}}_y \Phi^{(h_{d})}_{\cdot y}} \Phi^{(h_{d})}_{\cdot y}. \label{eqn:DN_update}
\end{gather}
Thus, if $\Phi^{(h_{d})}$ has been computed, then $\Phi^{(h_{d,y})}$ can be computed in $\OO(K^2)$ computations. In addition, if $\Phi^{(h_{d})}$, $D^{(h_{d})}$ and $N^{(h_{d})}$ have been computed, then $D^{(h_{d,y})}$ and $N^{(h_{d,y})}$ can be computed in $\OO(K)$ computations.}

\textcolor{black}{So, we can use \eqref{eqn:phi_update} and \eqref{eqn:DN_update} to implement Alg.~\ref{Alg:Widx_comp} in a more efficient manner. However, there is one additional step that needs to be handled, which we explain next.}

\textcolor{black}{As $h_{d+1} = \bar{g}^{({\cal P}_d \cup {\Gamma_{d+1}})}$, we get 
\begin{equation} \label{eqn:phi_update_new}
  I - \beta P^{(h_{d+1})} = \bigl(I - \beta P^{(h_d)}\bigr) + \beta \sum_{y \in \Gamma_{d+1}} e_y \rho^{\mathsf{T}}_y.
\end{equation}
Thus, $I - \beta P^{(h_{d+1})}$ is a rank-$|\Gamma_{d+1}|$ update of $I - \beta P^{(h_{d})}$. When $|\Gamma_{d+1}| > 1$, we can either 
sequentially apply equations \eqref{eqn:phi_update} and \eqref{eqn:DN_update} for all $y \in \Gamma_{d+1}$ or use the Woodbury formula\footnote{Suppose $A \in \mathbb{R}^{n\times n}$ is an invertible matrix and $U, V \in \mathbb{R}^{n \times m}$ are such that $A + U V^{\mathsf{T}}$ is invertible. Then the Woodbury formula is
\[ (A + UV^{\mathsf{T}})^{-1} = A^{-1} - A^{-1} U(I + V^{\mathsf{T}} U)^{-1} V^{\mathsf{T}} A^{-1}.\]} 
to compute $\Phi^{(h_{d+1})}$ and compute $D^{(h_{d+1})}$ and $N^{(h_{d+1})}$ using~\eqref{eqn:DN_PhiComp}. The complexity of sequentially applying the Sherman-Morrison formula is $\OO(|\Gamma_{d+1}|K^2)$ to compute $\Phi^{(h_{d+1})}$ and $\OO(|\Gamma_{d+1}| K)$ to compute $D^{(h_{d+1})}$ and $N^{(h_{d+1})}$. The complexity of using the Woodbury formula is $\OO(|\Gamma_{d+1}|^{2.807} + K^2)$ to compute $\Phi^{(h_{d+1})}$ and $\OO(K^2)$ to compute $D^{(h_{d+1})}$ and $N^{(h_{d+1})}$.} 

\textcolor{black}{We show the complete algorithm to efficiently compute the Whittle index in Algorithm~\ref{Alg:Widx_comp_eff}, where we use sequential application of Sherman-Morrison formula to compute $\Phi^{(h_{d+1})}$, $D^{(h_{d+1})}$ and $N^{(h_{d+1})}$.}

\begin{algorithm}[!t]
	\DontPrintSemicolon
	\SetKwInOut{Input}{input}
	\Input{RB~$(\mathcal{X}, \{0, 1\}, {P(a)}_{a \in \{0, 1\}}, c, x_0)$,
		discount factor $\beta$.} 
	Initialize $d = 0$, $\mathcal P_0 = \emptyset$, $h_0 = \boldsymbol{1}_{K}$. \;
    Compute $\Phi^{(h_0)} = (I - \beta P^{(h_0)})^{-1}$ and $[D^{(h_0)}\; N^{h_0}] = (1 - \beta) \Phi^{(h_0)} [c^{(h_0)}\; h_0]$ \;
	\While{$\mathcal{P}_d \neq \mathcal{X}$}{
		\ForAll{$y \in \mathcal{X}\setminus \mathcal{P}_d$}{
			Compute $D^{(h_{d, y})}$ and $N^{(h_{d, y})}$ using \eqref{eqn:DN_update}. \;
			Compute $\Lambda_{d, y}$ and $\mu_{d,y}(x)$ for all $x \in \Lambda_{d, y}$ using~\eqref{eq:mu}. \;
			Compute $\mu^*_{d,y} = \min_{x \in \Lambda_{d,y}} \mu_{d,y}(x)$. \;
		}
		Compute $\lambda_{d+1} = \min_{y \in {\cal X}\backslash {\cal P}_d} \mu^*_{d,y}$. \;
		Compute $\Gamma_{d+1} = \arg\min_{y \in {\cal X}\backslash {\cal P}_d} \mu^*_{d,y}$. \;
		Set $w(z) = \lambda_{d+1}$, $\forall z \in \Gamma_{d+1}$. \;
		Set ${\cal P}_{d+1} = {\cal P}_d \cup \Gamma_{d+1}$. \;
		Initialize $\Phi^{(h_{d+1})} = \Phi^{(h_{d})}$, $D^{(h_{d+1})} = D^{(h_{d})}$ and $N^{(h_{d+1})} = N^{(h_{d})}$. \;
		\ForAll{$z \in \Gamma_{d+1}$}{
			Compute $\Phi^{(h_{d+1,z})}$, $D^{(h_{d+1,z})}$ and $N^{(h_{d+1,z})}$ by using \eqref{eqn:phi_update} and \eqref{eqn:DN_update}. \;
			Update $\Phi^{(h_{d+1})} = \Phi^{(h_{d+1,z})}$, $D^{(h_{d+1})} = D^{(h_{d+1,z})}$ and $N^{(h_{d+1})} = N^{(h_{d+1,z})}$.
		}
		Set $d = d+1$. \;
	}
	\caption{Computing Whittle index of all states of an indexable RB}
	\label{Alg:Widx_comp_eff}
\end{algorithm}


\paragraph*{Some remarks}
\begin{enumerate}
  \item The idea of computing the index by iteratively sorting the states
    according to their index is commonly used in the algorithms to
    compute Gittins index; for example, the largest-remaining-index algorithm,
    the state-elimination algorithm, the triangularization algorithm, and the
    fast-pivoting algorithm use variations of this idea.
    See~\cite{chakravorty2014multi} for details.

  \item \textcolor{black}{The computational complexity of Algorithm~\ref{Alg:Widx_comp_eff} is $\OO(K^3)$, which can be characterized as follows. The algorithm starts with computing $\Phi^{(h_0)}$ which requires $\OO(K^{2.807})$ computations (using Strassen's algorithm) and $D^{(h_0)}$ and $N^{(h_0)}$ each of which requires $\OO(K^2)$ computations.
  Then, in the inner for loop, computing each of $D^{(h_{d, y})}$, $N^{(h_{d, y})}$ and $\mu^*_{d,y}$ requires $\OO(K)$ computations and the inner loop is executed $|{\cal X} \backslash {\cal P}_d|$ times.
  Afterwards, updating $\Phi^{(h_{d+1})}$, $D^{(h_{d+1})}$ and $N^{(h_{d+1})}$ requires $\OO(|\Gamma_{d+1}|K^2)$, $\OO(|\Gamma_{d+1}|K)$ and $\OO(|\Gamma_{d+1}|K)$ computations, repectively.
  Therefore, the computational complexity of the algorithm is
  \begin{align*}
    \hskip 1em & \hskip -1em
    \OO(K^{2.807}) + \OO(K^2) + \sum_{d = 1}^{K_D} \left( \OO(|{\cal X} \backslash {\cal P}_d| K) + \OO(|\Gamma_{d+1}| K^2) + \OO(2 |\Gamma_{d+1}| K) \right) \\
    &\le \OO(K^{2.807}) + \sum_{d = 1}^{K_D} \OO(K^2) + \OO\biggl(\biggl[ \sum_{d = 1}^{K_D} |\Gamma_{d+1}| \biggr] K^2\biggr)  \\
    &\le \OO(K^{2.807}) + \OO(K^3) + \OO(K^3) 
    \le \OO(K^3),
  \end{align*}
where the first inequality uses the fact that $|\mathcal{X}\setminus\mathcal{P}_d| \le K$ and the second inequality uses the fact that $\sum_{d=1}^{K_D}| \Gamma_{d+1}| = K$.} 

    
  \item \textcolor{black}{Note that Algorithm~\ref{Alg:Widx_comp_eff} computes the Whittle index exactly. In contrast, using binary search~\cite{akbarzadeh2019dynamic} computes the Whittle index approximately. Let $C_{\max}$ and $C_{\min}$ denote the upper and lower bound on the per-step cost. Then, we know that for any state $x$, $w(x) \in [ C_{\min}, C_{\max} ]$. Now, suppose we want to compute the Whittle index to an accuracy of $\delta$. Then the interval $[C_{\min}, C_{\max}]$ needs to be divided into $\log_2( (C_{\max} - C_{\min})/\delta)$ steps. For each step of the binary search, we need to solve the dynamic program~\eqref{eqn:bllmn_vf}. Solving the dynamic program exactly using policy iterations has a complexity of $\OO(K^3)$. Solving it approximately using value iteration has a complexity of $\OO(N_{\textup{VI}}K^2)$, where $N_{\textup{VI}}$ is the number of iterations for value iteration (see~\cite{feinberg2020complexity} for bounds on $N_{\textup{VI}}$). Note that the binary search needs to be repeated for each state. Thus, using binary search to compute Whittle index to an accuracy of $\delta$ has a complexity $\OO(\log_2( (C_{\max} - C_{\min})/\delta) N_{\textup{VI}} K^4)$ if the dynamic program at each step is solved exactly and has a complexity of $\OO(\log_2( (C_{\max} - C_{\min})/\delta) N_{\textup{VI}} K^3)$ if the dynamic program at each step is solved approximately.}
    
\end{enumerate}

\subsection{Discussion on PCL-indexability} \label{subsec:pcl}

As mentioned earlier, an algorithm very similar to Alg.~\ref{Alg:Widx_comp}
was proposed in~\cite{nino2007dynamic} for computing the Whittle index for RBs
that satisfy a technical condition known as PCL-indexability. The analysis 
in~\cite{nino2007dynamic} is done under the assumption that the system starts
from a designated start state distribution~$\pi_0$. For any policy $g$, define
$\mathsf{N}^{(g)} = \sum_{x \in \mathcal{X}} N^{(g)}(x) \pi_0(x)$ and define
$\mathsf{D}^{(g)} = \sum_{x \in \mathcal{X}} D^{(g)}(x) \pi_0(x)$. Let $x_1,
\dots, x_K$ be a permutations of state space such that the corresponding Whittle
indices are \textcolor{black}{non-decreasing}: $\lambda_1 \le \dots \le \lambda_K$. For any $k \in \{1, \dots, K\}$,
let $\bar{\mathcal{P}}_k$ denote the set $\{ x_1, \dots, x_k \}$.

Now for any $k \in \{1, \dots, K\}$, and all states $y \in \mathcal{X}
\setminus \bar{\mathcal{P}}_k$, define $\bar h_k = \bar
g^{(\bar{\mathcal{P}}_k)}$, $\bar h_{k,y} = \bar g^{(\bar{\mathcal{P}}_k \cup
\{y\})}$, and define
\begin{equation}\label{eq:MPI}
  \bar \mu_{k,y} = \frac{ \mathsf{D}^{(\bar h_{k,y})} - \mathsf{D}^{(\bar h_k)} }
  {\mathsf{N}^{(\bar h_k)} - \mathsf{N}^{(\bar h_{k,y})}}.
\end{equation}

In \cite{nino2007dynamic} an algorithm, called
the adaptive greedy algorithm, is presented to iteratively identify the sets
$\bar{\mathcal{P}}_k$ and compute the corresponding Whittle indices. This
algorithm is shown in Alg.~\ref{Alg:PCL}.

\begin{algorithm}[!t]
  \DontPrintSemicolon
  \SetKwInOut{Input}{input}
  \Input{RB~$(\mathcal{X}, \{0, 1\}, {P(a)}_{a \in \{0, 1\}}, c, x_0)$,
  discount factor $\beta$.} 
  Initialize $k = 0$ and $\mathcal P_0 = \emptyset$ \;
  \While{$k \neq K$}{
    \ForAll{$y \in \mathcal{X}\setminus \mathcal{P}_k$}{
    Compute $\bar \mu_{d,y}$ using~\eqref{eq:MPI}\;
  }
  Pick $x_{k+1} \in \arg \min_{y \in {\cal X}\backslash \bar{\cal P}_k} \bar
  \mu_{k,y}$ \;
  Set $w(s_{k+1}) = \min_{y \in {\cal X}\backslash \bar{\cal P}_k} \bar \mu_{k,y}$  
  and $\bar{\cal P}_{k+1} = \bar{\cal P} \cup \{x_k\}$ \;
	$k = k+1$ \;
  }
  \caption{The Adaptive Greedy Algorithm of \cite{nino2007dynamic}}
  \label{Alg:PCL}
\end{algorithm}

A RB to be PCL-indexable \cite{nino2007dynamic} if it satisfies the following conditions:
\begin{enumerate}
  \item For any ${\cal S} \subseteq {\cal X}$ and $y \in {\cal
    X}\backslash{\cal S}$, we have
    $\mathsf{N}^{(\bar{g}^{({\cal S})})} - \mathsf{N}^{(\bar{g}^{({\cal S} \cup
    \{y\})})} > 0$.
	\item The sequence of index values produced by the adaptive greedy
      algorithm  is monotonically non-decreasing.
\end{enumerate}

Finally, the following result is established:
\begin{theorem}[Theorem 1 of \cite{nino2007dynamic}] \label{Thm:PCL}
	A PCL-indexable RB is indexable and the adaptive greed algorithm gives its Whittle indices.
\end{theorem}

The main differences between our result and~\cite{nino2007dynamic} are as
follows:
\begin{enumerate}
  \item An implication of the first condition in the definition of PCL
    indexability is that the denominator in~\eqref{eq:MPI} is never zero. In
    contrast, we do not impose such a restriction and work with the non-empty subset of
    states for which the denominator in~\eqref{eq:mu} is non-zero.
  \item In Alg.~\ref{Alg:PCL}, the sets $\{ \bar{\cal P}_k \}_{k=1}^{K}$ are
    constructed by adding states one-by-one, even when $\bar \mu_{k,y}$ has
    multiple argmins. In contrast, in Alg.~\ref{Alg:Widx_comp}, the sets $\{
    \mathcal{P}_d \}_{d=1}^{K_D}$ are constructed by adding all states which
    have the same Whittle index at once.
  \item In Alg.~\ref{Alg:PCL}, one has to check that the indices are generated
    in a nondecreasing order (which is the second condition of
    PCL-indexability). In contrast, Alg.~\ref{Alg:Widx_comp}, the indices are
    always generated in an increasing order and, therefore, 
    condition~2 of PCL-indexability is always satisfied. 
  \item Finally, Theorem~\ref{Thm:PCL} only guarantees that Alg.~\ref{Alg:PCL}
    computes the Whittle index for RBs which satisfy PCL-indexability.
    Moreover, the second condition in PCL-indexability can only be checked
    after running Alg.~\ref{Alg:PCL}. In contrast, Theorem~\ref{Thm:Widx}
    guarantees that Alg.~\ref{Alg:Widx_comp} computes the Whittle index for
    all indexable RBs.
\end{enumerate}

We conclude this discussion by revisiting an example from \cite{nino2007dynamic} which is
an indexable RB but not PCL-indexable. For this
example, $\mathcal{X} = \{1, 2, 3\}$, the transition matrices are
\( P(0) = \left[\begin{smallmatrix} 0.3629 & 0.5028 & 0.1343 \\ 0.0823 & 0.7534 & 0.1643 \\ 0.2460 & 0.0294 & 0.7246
  \end{smallmatrix}\right]\)
  and
\(
P(1) = \left[\begin{smallmatrix} 0.1719 & 0.1749 & 0.6532 \\ 0.0547 & 0.9317 & 0.0136 \\ 0.1547 & 0.6271 & 0.2182
\end{smallmatrix}\right] \), the per-step cost is
$c(x,0) = 0$ for all $x \in \mathcal{X}$ and $c(1, 1) = -0.44138$, $c(2, 1) = -0.8033$, $c(3, 1) = -0.14257$ , and $\beta = 0.9$.  and the corresponding Whittle indices are $[0.18, 0.8, 0.57]$. 

This model is not PCL-indexable since if $g = [1, 1, 0]$ and $h = [0, 1, 0]$, then $N^{(g)} = [5.66, 8.24, 4.23]$ and $N^{(h)} = [6.65, 8.59, 4.88]$. Therefore, for any initial state distribution~$\pi_0$, $\mathsf{N}^{(g)} < \mathsf{N}^{(h)}$. 

\textcolor{black}{However, as the problem is indexable, we can still apply Alg.~\ref{Alg:Widx_comp} to compute Whittle indices without any limitations. The steps are as follows:
\begin{enumerate}
	\item Initialize $d=0$ and have ${\cal P}_0 =
		\emptyset$. Thus ${\bar g}^{({\cal P}_0)} = [1, 1, 1]$ and we compute $N^{({\bar g}^{({\cal
					P}_0)})} = [10, 10, 10]$ and $D^{({\bar g}^{({\cal P}_0)})} = [-6.43, -7.43, -6.51]$. 
	\item There are three possibilities for $y \in \mathcal{X}
		\setminus \mathcal{P}_0 = \{1,2,3\}$:
	\begin{itemize}
		\item For $y = 1$, $h_{0,1} = [0, 1, 1]$. We compute 
			$N^{(h_{0,1})} = [7.88, 9.29, 9.13]$ and $D^{(h_{0,1})} = [-6.05, -7.30, -6.35]$. Therefore, $\Lambda_{0,1} =
			\{x \in \mathcal{X} : N^{(\bar g(\mathcal{P}_0)}(x) \neq
			N^{(h_{0,1})}(x) \} = \{1, 2, 3\}$. Now for each $x \in
			\Lambda_{0,1}$, we compute $\mu_{0,1}(1) = \mu_{0,1}(2) =
			\mu_{0,1}(3) = 0.18$. Therefore, $\mu^*_{0,1} = 0.18$.
		\item For $y = 2$, $h_{0,2} = [1, 0, 1]$. We compute 
			$N^{(h_{0,2})} = [4.58, 2.93, 4.10]$ and $D^{(h_{0,2})} = [-1.27, -0.7, -0.89]$. Therefore, $\Lambda_{0,2} =
			\{x \in \mathcal{X} : N^{(\bar g^{(\mathcal{P}_0)}}(x) \neq
			N^{(h_{0,2})}(x) \} = \{1, 2, 3\}$. Now for each $x \in
			\Lambda_{0,2}$, we compute $\mu_{0,2}(1) = \mu_{0,2}(2) =
			\mu_{0,2}(3) = 0.95$. Therefore, $\mu^*_{0,2} = 0.95$.
		\item For $y = 3$, $h_{0,3} = [1, 1, 0]$. We compute 
			$N^{(h_{0,3})} = [5.66, 8.24, 4.23]$ and $D^{(h_{0,3})} = [-3.64, -6.30, -2.79]$. Therefore, $\Lambda_{0,3} =
			\{x \in \mathcal{X} : N^{(\bar g^{(\mathcal{P}_0)}}(x) \neq
			N^{(h_{0,3})}(x) \} = \{1, 2, 3\}$. Now for each $x \in
			\Lambda_{0,3}$, we compute $\mu_{0,3}(1) = \mu_{0,3}(2) =
			\mu_{0,3}(3) = 0.64$. Therefore, $\mu^*_{0,3} = 0.64$.
	\end{itemize}
	Now $\lambda_1 = \min \{ \mu^*_{0,1}, \mu^*_{0,2},
		\mu^*_{0,3} \} = 0.18$. Therefore, ${\cal P}_1 = \{1\}$, $w(1) = 0.18$,
		${\bar g}^{({\cal P}_1)} = [0, 1, 1]$. We have already computed $N^{({\bar g}^{({\cal P}_1)})} = [7.88,
		9.29, 9.13]$ and $D^{({\bar g}^{({\cal P}_1)})} = [-6.05, -7.30, -6.35]$.
	\item There are two possibilities for $y \in \mathcal{X}
		\setminus \mathcal{P}_1 = \{2, 3\}$:
	\begin{itemize}
		\item For $y = 2$, $h_{1,2} = [0, 0, 1]$. We compute 
			$N^{(h_{1,2})} = [1.48, 1.52, 2.57]$ and $D^{(h_{1,2})} = [-0.21, -0.22, -0.37]$. Therefore, $\Lambda_{1,2} =
			\{x \in \mathcal{X} : N^{(\bar g^{(\mathcal{P}_1)}}(x) \neq
			N^{(h_{1,2})}(x) \} = \{1, 2, 3\}$. Now for each $x \in
			\Lambda_{1,2}$, we compute $\mu_{1,2}(1) = \mu_{1,2}(2) =
			\mu_{1,2}(3) = 0.91$. Therefore, $\mu^*_{1,2} = 0.91$.
		\item For $y = 3$, $h_{1,3} = [0, 1, 0]$. We compute 
			$N^{(h_{1,3})} = [6.65, 8.59, 4.88]$ and $D^{(h_{1,3})} = [-1.22, -0.66, -0.83]$. Therefore, $\Lambda_{1,3} =
			\{x \in \mathcal{X} : N^{(\bar g^{(\mathcal{P}_1)}}(x) \neq
			N^{(h_{1,3})}(x) \} = \{1, 2, 3\}$. Now for each $x \in
			\Lambda_{1,3}$, we compute $\mu_{1,3}(1) = \mu_{1,3}(2) =
			\mu_{1,3}(3) = 0.57$. Therefore, $\mu^*_{1,3} = 0.57$.
	\end{itemize}
	Now $\lambda_2 = \min \{ \mu^*_{1,2},
		\mu^*_{1,3} \} = 0.57$. Therefore, ${\cal P}_2 = \{1, 3\}$, $w(3) = 0.57$,
		${\bar g}^{({\cal P}_2)} = [0, 1, 0]$. We have already computed $N^{({\bar g}^{({\cal P}_2)})} = [6.65, 8.59, 4.88]$ and $D^{({\bar g}^{({\cal P}_2)})} = [-1.22, -0.66, -0.83]$.
	\item There is only one possibility for $y \in \mathcal{X}
		\setminus \mathcal{P}_2 = \{2\}$:
	\begin{itemize}
		\item For $y = 3$, $h_{3,2} = [0, 0, 0]$, $N^{(h_{3,2})} = [0, 0, 0]$ and $D^{(h_{3,2})} = [-0.21, -0.22, -0.37]$. Therefore, $\Lambda_{3,2} =
			\{x \in \mathcal{X} : N^{(\bar g^{(\mathcal{P}_2)}}(x) \neq
			N^{(h_{3,2})}(x) \} = \{1, 2, 3\}$. Now for each $x \in
			\Lambda_{3,2}$, we compute $\mu_{3,2}(1) = \mu_{3,2}(2) =
			\mu_{3,2}(3) = 0.8$. Therefore, $\mu^*_{3,2} = 0.8$.
	\end{itemize}
	Now $\mu^*_{3,2} = 0.57$. Therefore, ${\cal P}_3 = \{1, 2, 3\}$ and $w(2) = 0.8$.
\end{enumerate}
Finally, the Whittle indices are $[0.18, 0.8, 0.57]$.}

\section{Some special cases}\label{sec:special}

In this section, we refine the results developed in this paper 
to some special cases. 

\subsection{Restless bandits with optimal threshold-based policy}

Consider a RB $(\mathcal{X}, \{0, 1\}, \allowbreak \{P(a)\}_{a \in \{0, 1\}},
c, x_0)$ where the state space $\mathcal{X}$ is a totally ordered set\textcolor{black}{, i.e., ${\cal X} = \{1, \ldots, K\}$}.
Let $\mathcal{X}_0 = \{0, \dots, K\}$ and 
let $\mathcal{X}_{\ge \ell}$ denotes the set of states greater than or equal
to state~$\ell$ and $\mathcal{X}_{\le \ell}$ denotes the set of states less
than or equal to state~$\ell$. We suppose that the model satisfies the following assumption:
\begin{itemize}
  \item[\textup{(P)}] There exists a non-decreasing family of thresholds
    $\{\ell_\lambda\}_{\lambda \in \mathbb{R}}$, $\ell_\lambda \in
    \mathcal{X}_0$, such that the threshold policy $g^{(\ell_\lambda)}$ is
    optimal for Problem~\ref{prob:decompose} with activation
    cost~$\lambda$.
\end{itemize}
Several models where (P) holds have been considered in the literature~\cite{ansell2003whittle, glazebrook2009index,
avrachenkov2013congestion, glazebrook2013monotone, wang2019whittle,
akbarzadeh2019restless}. A key implication of property~(P) is the following: 
\begin{lemma}\label{lem:threshold}
  Suppose a RB defined on a totally ordered state space satisfies property~(P).
  Then, the restless bandit is indexable and the Whittle index $w(\ell)$ is
  non-decreasing in~$\ell \in {\cal X}$. 
\end{lemma}
\begin{proof}
  Note that property~(P) implies that the passive set
  $\Pi_\lambda = \{ x \in \mathcal{X} : g_\lambda(x) = 0 \} = \mathcal{X}_{\le
  \ell_\lambda}$, which is increasing in~$\lambda$. Hence the RB is indexable.
  Moreover, for any state~$\ell$, the Whittle index $w(\ell)$ is the smallest value
  of~$\lambda$ such that $\ell_\lambda = \ell$. Therefore, by Property~(P),
  $w(\ell)$ is non-decreasing in~$\ell$.
\end{proof}

\textcolor{black}{%
As in Section~\ref{sec:computation}, we assume that there are $K_D (\leq K)$ distinct Whittle indices given by $\Lambda^* = \{\lambda_1, \dots, \lambda_{K_D}\}$ where $\lambda_1 < \lambda_2 < \dots \lambda_{K_D}$. We also let $\lambda_0 = -\infty$ and for any $d \in \{0, \dots, K_D\}$, let $\mathcal{P}_d = \{ x \in \mathcal X : w(x) \le \lambda_d \}$. As stated in the proof of Lemma~\ref{lem:threshold} property~(P) implies that $\mathcal{P}_d = \mathcal{X}_{\le \ell_{\lambda_d}}$. Therefore, $\Gamma_{d+1} = \{\ell_{\lambda_d} + 1, \dots, \ell_{\lambda_{d+1}}\}$. Thus, Theorem~\ref{Thm:Widx} simplifies as follows:
\begin{corollary}\label{Cor:P}
  Suppose a RB defined on a totally ordered state space satisfies property~(P). Then, the following properties hold:
	\begin{enumerate}
      \item For any $y \in \Gamma_{d+1}$,
        the set $\Lambda_{d, y}$ is non-empty.
		\item For any $x \in \Lambda_{d, y}$, $\mu_{d,y}(x) \geq \lambda_{d+1}$ with equality if and only if $y \in  \Gamma_{d+1}$.
	\end{enumerate}
\end{corollary}}


\textcolor{black}{Thus, based on Corollary~\ref{Cor:P}, for models that satisfy property~(P), we can simplify Algorithm~\ref{Alg:Widx_comp_eff} as shown in Algorithm~\ref{Alg:Widx_comp_threshold}. Instead of computing $\mu^*_{d,y}$ for all $y \in \mathcal{X} \setminus \mathcal{P}_d$, we can compute it sequentially and break the loop when $\mu^*_{d,y} \neq \lambda_{d+1}$. Note that this simplification does not change the asymptotic complexity of the algorithm, which is still $\OO(K^3)$.}

\begin{algorithm}[!t]
	\DontPrintSemicolon
	\SetKwInOut{Input}{input}
	\Input{RB~$(\mathcal{X}, \{0, 1\}, {P(a)}_{a \in \{0, 1\}}, c, x_0)$, discount factor $\beta$.} 
    Initialize $d = 0$, $\ell = 0$, $h_0 = \boldsymbol{1}_{K}$. \;
    Compute $\Phi^{(h_0)} = (I - \beta P^{(h_0)})^{-1}$, $[D^{(h_0)}\; N^{h_0}] = (1 - \beta) \Phi^{(h_0)} [c^{(h_0)}\; h_0]$ \;
	\While{$\ell \le K$}{
      \ForAll{$y \in \{\ell + 1, \dots, K\}$}{
			Compute $D^{(h_{d, y})}$ and $N^{(h_{d, y})}$ using \eqref{eqn:DN_update}. \;
			Compute $\Lambda_{d, y}$ and $\mu_{d,y}(x)$ for all $x \in \Lambda_{d, y}$ using~\eqref{eq:mu}. \;
			Compute $\mu^*_{d,y} = \min_{x \in \Lambda_{d,y}} \mu_{d,y}(x)$. \;
            \eIf{$y = \ell + 1$}{
              Set $\lambda_{d+1} = \mu^*_{d,y}$ and
              $\Gamma_{d+1} = \{ y \}$ \;
              Set $w(y) = \lambda_{d+1}$ \;
            }{
              \eIf{$\lambda_{d+1} = \mu^*_{d,y}$}
              {Update $\Gamma_{d+1} = \Gamma_{d+1} \cup \{ y \}$ \;
              Set $w(y) = \lambda_{d+1}$ \;
              }
              {Set $\ell = y$ \; break}
            }
		}
		Initialize $\Phi^{(h_{d+1})} = \Phi^{(h_{d})}$, $D^{(h_{d+1})} = D^{(h_{d})}$ and $N^{(h_{d+1})} = N^{(h_{d})}$. \;
		\ForAll{$z \in \Gamma_{d+1}$}{
			Compute $\Phi^{(h_{d+1,z})}$, $D^{(h_{d+1,z})}$ and $N^{(h_{d+1,z})}$ by using \eqref{eqn:phi_update} and \eqref{eqn:DN_update}. \;
			Update $\Phi^{(h_{d+1})} = \Phi^{(h_{d+1,z})}$, $D^{(h_{d+1})} = D^{(h_{d+1,z})}$ and $N^{(h_{d+1})} = N^{(h_{d+1,z})}$.
		}
		Set $d = d+1$. \;
	}
    \caption{Whittle index for RB with optimal threshold-based policy}
	\label{Alg:Widx_comp_threshold}
\end{algorithm}

\textcolor{black}{
  \begin{remark}
    Note that if the model satisfies additional assumptions such that it is known upfront that no two states have the same Whittle index, then we don't need the inner for loop (over~$y$) in Algorithm~\ref{Alg:Widx_comp_threshold}, and can simply compute the Whittle index of state~$\ell$ as
    \[
      w(\ell) = \min \frac{D^{(\bar g^{(\mathcal{X}_{\le \ell + 1})})}(x) - D^{(g^{(\mathcal{X}_{\le \ell })})}(x)}
      {N^{(\bar g^{(\mathcal{X}_{\le \ell})})}(x) - N^{(g^{(\mathcal{X}_{\le \ell + 1 })})}(x)},
    \]
    where the minimum is over all $x$ such that the denominator is no zero.
  \end{remark}
}

In the next section, we present a new model called
\emph{stochastic monotone bandits}, which may be considered as a
generalization of monotone bandits~\cite{glazebrook2013monotone,
ansell2003whittle, avrachenkov2013congestion}, and show that these models
satisfy property~(P).

\subsection{Stochastic monotone bandits}

We say that the RB is \emph{stochastic monotone} if it satisfies the
following conditions.
\begin{enumerate}
    \setlength{\itemsep}{0.25\baselineskip}
  \setlength{\parskip}{0pt}
  \item[\textup{(D1)}] For any $a \in \{0, 1\}$, $P(a)$ is stochastically monotone, i.e., for any $x, y \in \mathcal{X}$ such that $x < y$, we have $\sum_{w \in \mathcal X_{\ge z}} P_{xw}(a) \leq \sum_{w \in \mathcal X_{\ge z}} P_{yw}(a)$ for any $z \in \mathcal{X}$.
  \item[\textup{(D2)}] For any $z \in \mathcal{X}$, $S_{zx}(a) := \sum_{w \in
    \mathcal X_{\ge z}} P_{xw}(a)$ in submodular\footnote{Given ordered sets
      $\mathcal{X}$ and $\mathcal{Y}$, a function $f: \mathcal{X} \times
      \mathcal{Y} \to \mathbb{R}$ is called submodular if for any $x_1, x_2
      \in \mathcal{X}$ and $y_1, y_2 \in \mathcal{Y}$ such that $x_2 \geq x_1$
      and $y_2 \geq y_1$, we have $f(x_1, y_2) - f(x_1, y_1) \geq f(x_2, y_2)
    - f(x_2, y_1)$.} in $(x, a)$.
  \item[\textup{(D3)}] For any $a \in \{ 0, 1 \}$, $c(x, a)$ is non-decreasing in $x$. 
  \item[\textup{(D4)}] $c(x, a)$ is submodular in $(x, a)$.
\end{enumerate}


For ease of notation, for any $\ell \in \mathcal{X}_0$, 
we let $g^{(\ell)} = \bar g^{(\mathcal X_{\le \ell})}$ denote a policy with
threshold~$\ell$ (where $\bar{g}^{(\cal S)}$ is as defined
in~\eqref{eqn:gbar}). 

\begin{lemma} \label{lem:monotone}
  A stochastic monotone RB satisfies the following properties:
  \begin{enumerate}
    \item For any $\lambda \in \mathbb{R}$, there exists a threshold
      $\ell_\lambda \in \mathcal{X}^*$ such that the thershold
      policy~$g^{(\ell_\lambda)}$ is optimal for Problem~\ref{prob:decompose}.
      If there are multiple such thresholds, we use $\ell_\lambda$ to denote
      the largest threshold. 
    \item If, for any $x \in \mathcal{X}$, $N^{(g^{(\ell)})}(x)$ is
      non-increasing in $\ell$, then $\ell_\lambda$ is non-decreasing with
      $\lambda$. Therefore, the model satisfies property~(P) 
      and is, therefore, indexable.
  \end{enumerate}
\end{lemma}
\begin{proof}
  For the first part, we note that 
  conditions \textup{(D1)}--\textup{(D4)} are the same as the properties
  of~\cite[Theorem 4.7.4]{puterman2014markov}, which implies that there exists a threshold based
  \endgraf
  For the second part, we first show that for any $\ell
  \in \mathcal{X}^*$, $J^{(g^{(\ell)})}_\lambda(x)$ is submodular in
  $(\ell,\lambda)$ for all $x \in \mathcal{X}$. In particular, for any $k <
  \ell$, we have
  \[
    J^{(g^{(\ell)})}_\lambda(x) - J^{(g^{(k)})}_\lambda(x) = 
    D^{(g^{(\ell)})}_\lambda(x) - D^{(g^{(k)})}_\lambda(x) +
    \lambda( N^{(g^{(\ell)})}_\lambda(x) - N^{(g^{(k)})}_\lambda(x)).
  \]
  Now (D5) implies that the difference $J^{(g^{(\ell)})}_\lambda(x) -
  J^{(g^{(k)})}_\lambda(x)$ is non-increasing in $\lambda$. Therefore,
  $J^{(g^{(\ell)})}_\lambda(x)$ is submodular in $(\ell,\lambda)$. Consequently,
  from~\cite[Theorem 2.8.2]{puterman2014markov}, $\ell_\lambda = \max \{ \ell' \in
  \arg\min_{\ell \in \mathcal{X}^*} J^{(g^{(\ell)})}_\lambda(x) \}$ is
  non-decreasing in~$\lambda$.
\end{proof}

\subsection{Restless bandits with controlled restarts}

Consider restless bandits with controlled restarts (i.e., models where
$P_{xy}(1)$ does not depend on $x$). By Proposition~\ref{prop:1}c, such models
are indexable. In this section, we explain how to simplify the computation of the
Whittle index for such models. For ease of notation, we use $P_{xy}$ to denote
$P_{xy}(0)$ and $Q_y$ to denote $P_{xy}(1)$. 

Define
$
  \mathsf{D}^{(g)} = \sum_{x \in \mathcal{X}} Q_x D^{(g)}(x)
  \quad\text{and}\quad
  \mathsf{N}^{(g)} = \sum_{x \in \mathcal{X}} Q_x N^{(g)}(x).
$
Now, following the discussion of Sec.~\ref{sec:computation}, we can show that
the result of
Theorem~\ref{Thm:Widx} continues to holds when $\mu_{d,y}$ is replaced by 
\textcolor{black}{
\[ \hat{\mu}_{d,y} = \dfrac{\mathsf{D}^{(h_y)} - \mathsf{D}^{(\bar{g}^{({\cal P}_d)})}} {\mathsf{N}^{(\bar{g}^{({\cal P}_d)})} - \mathsf{N}^{(h_y)}}. \]
Therefore, we can replace $\mu_{d,y}(x)$ in
Algorithm~\ref{Alg:Widx_comp} by $\hat{\mu}_{d,y}$.} Our key result for this
section is  $\mathsf{D}^{(g)}$ and $\mathsf{N}^{(g)}$ can be computed efficiently for models with controlled
restarts.

For that matter, given any policy $g$, let $\tau_g$ denote the hitting time of
the set $\Pi^{(g)} = \{x \in \mathcal{X} : g(x) = 1 \}$.  Let 
\[
  \mathsf{L}^{(g)} \coloneqq \EXP\Bigl[
  \sum_{t=0}^{\tau_g} \beta^t c(X_t, g(X_t)) \Bigm| X_0 \sim Q \Bigr] 
  \quad\text{and}\quad
  \mathsf{M}^{(g)} \coloneqq \EXP\Bigl[
  \sum_{t=0}^{\tau_g} \beta^t \Bigm| X_0 \sim Q \Bigr]
\]
denote the expected discounted cost and expected discounted time for hitting
$\Pi^{(g)}$ starting with an initial state distribution of $Q$. 
Then, using ideas from renewal theory, we can show the following.
\begin{theorem}\label{Thm:DNC_funcs}
  For any policy $g$,
	\begin{equation*}
	\mathsf{D}^{(g)}  = \dfrac{\mathsf{L}^{(g)}}{\mathsf{M}^{(g)}}
	\quad\text{and}\quad
	\mathsf{N}^{(g)} = \dfrac{1}{\beta \mathsf{M}^{(g)}} - \dfrac{1-\beta}{\beta}. 
	\end{equation*}
\end{theorem}
\begin{proof}
  The proof follows from standard ideas in renewal theory. By strong Markov property, we have
  \begin{align}
    \mathsf{D}^{(g)} &= \EXP \Big[ (1-\beta) \sum_{t = 0}^{\tau_g} \beta^{t} c(X_t, g(X_t)) 
    + \beta^{\tau_g+1} \mathsf{D}^{(g)} \Bigm| X_0 \sim Q \Big] \nonumber \\
    &= (1-\beta) \mathsf{L}^{(g)} + \EXP[ \beta^{\tau_g+1} | X_0 \sim Q ] \mathsf{D}^{(g)}. \label{eqn:D_exp}
  \end{align}
  Using $\mathsf{M}^{(g)}$ definition, we have $\EXP[ \beta^{\tau_g+1} | X_0 \sim Q ] = 1 - (1-\beta) \mathsf{M}^{(g)}$.
  Substituting this in~\eqref{eqn:D_exp} and rearranging the terms we get $\mathsf{D}^{(g)} = \mathsf{L}^{(g)}/\mathsf{M}^{(g)}$.

  For $\mathsf{N}^{(g)}$, by strong Markov property we have
  \begin{align*}
  \mathsf{N}^{(g)} &= \EXP\Big[ (1-\beta) \beta^{\tau_g} + \beta^{\tau_g+1} \mathsf{N}^{(g)} \Bigm| X_0 \sim Q \Big] \nonumber \\
  &= \EXP[ \beta^{\tau_g} | X_0 \sim Q ] (1-\beta + \beta \mathsf{N}^{(g)}) = \dfrac{ 1-(1-\beta)\mathsf{M}^{(g)} }{\beta}(1-\beta + \beta \mathsf{N}^{(g)}).
\end{align*}
Therefore, we get $\mathsf{N}^{(g)} = \bigl(1 - (1-\beta)
\mathsf{M}^{(g)}\bigr) / \beta \mathsf{M}^{(g)}$.
\end{proof}
 
Given any policy $g$, we can efficiently compute $\mathsf{L}^{(g)}$ and
$\mathsf{M}^{(g)}$ using standard formulas for truncated Markov chains. For
any vector $v$, let $v^{(g)}$ denote the vector with components indexed by the
set $\{ x \in \mathcal{X} : g(x) = 0 \}$ and $\tilde v^{(g)}$ denote the
remaining components. For example, if $\mathcal{X} = \{1, 2, 3, 4\}$, $g = (1,
0, 1, 0)$, and $v = [1, 2, 3, 4]$, then $v^{(g)} = (2, 4)$ and $\tilde v^{(g)}
= (1,3)$. Similarly, for any square matrix $Z$, let $Z^{[g]}$ denote the
square sub-matrix corresponding to elements $\{ x \in \mathcal{X} : g(x) = 0
\}$, and $\tilde Z^{[g]}$ denote the sub-matrix with rows $\{x \in \mathcal{X}
: g(x) = 0 \}$ and columns $\{x \in \mathcal{X} : g(x) = 1 \}$. 
\footnote{For example, if $g = [1, 0, 1, 0]$ and if $Z = \left[\begin{smallmatrix} 1 & 2 & 3 & 4 \\ 5 & 6 & 8 & 8 \\ 9 & 10 & 11 & 12 \\ 13 & 14 & 15 & 16 \end{smallmatrix}\right]$, then $Z^{[g]} = \left[\begin{smallmatrix} 6 & 8 \\ 14 & 16 \end{smallmatrix}\right]$ and $Z^{[g]} = \left[\begin{smallmatrix} 5 & 8 \\ 13 & 15 \end{smallmatrix}\right]$.}.
Then, from standard formulas for truncated Markov chains, we have the following.
\begin{proposition}
  For any policy $g$, let $c_0$ and $c_1$ denote column vectors corresponding
  to $c(\cdot, 0)$ and $c(\cdot, 1)$. Then,
  \begin{align*}
    \mathsf L^{(g)} &= Q^{(g)} (I - \beta P^{[g]})^{-1}
    (c_0^{(g)} + \beta \tilde P^{[g]} \tilde c^{(g)}_1) 
    + \tilde Q^{(g)} \tilde c^{(g)}_1, 
    \\
    \mathsf M^{(g)} &= Q^{(g)} (I - \beta P^{[g]})^{-1}
    (\mathbf{1}^{(g)} + \beta \tilde P^{[g]} \tilde{\mathbf{1}}^{(g)}) 
    + \tilde Q^{(g)} \tilde {\mathbf{1}}^{(g)}. 
  \end{align*}
\end{proposition}
This gives us an efficient method to compute $\mathsf{L}^{(g)}$ and
$\mathsf{M}^{(g)}$, which can in turn be used to compute $\mathsf{D}^{(g)}$
and $\mathsf{N}^{(g)}$ and used in a modified version of
Algorithm~\ref{Alg:Widx_comp} as explained.

\section{Numerical Experiments} \label{sec:simulation}

In this section, we evaluate how well the Whittle index policy (\textsc{wip})
performs compared to the optimal policy (\textsc{opt}) as well as to a
baseline policy known as the myopic policy (\textsc{myp}) (shown in
Algorithm~\ref{alg:MYP}). The code is also available\footnote{\url{https://codeocean.com/capsule/8680851/tree/v1}}.
\begin{algorithm}[!t]
  \DontPrintSemicolon
  \SetKwInOut{Input}{input}
  \Input{Set $\mathcal{N}$ of arms; arms $m$ to be activated}
  \ForEach{time $t$}{
    let $\ell = 0$, $\mathcal{M} = \emptyset$, and $\mathcal{Z} = \mathcal{N}$. \;
    \ForEach{$\ell \in \{0, \ldots, m\}$}{
      $i^*_\ell \in \arg\min_{i \in \mathcal{Z}} \sum_{j \in \mathcal{Z}\setminus \{i\}}
      \{c^j(X^j_t, 0) + c^i(X^i_t, 1) \}$ 
      \tcp*{Pick any arg min}
      let $\mathcal{M} = \mathcal{M} \cup \{i^*_\ell\}$,
      $\mathcal{Z} = \mathcal{Z} \setminus \{i^*_\ell\}$ \; 
    }
    Activate arms in $\mathcal{Z}$\;
  }
  \caption{Myopic Heuristic} \label{alg:MYP}
\end{algorithm}

\subsection{Experimental Setup} \label{sebsec:setup}

In our experiments, we consider restart bandits with $P(1)= [\mathbf{1},
\mathbf{0}, \dots, \mathbf{0}]$. There are two other components of the model:
The transition matrix $P(0)$ and the cost function~$c$. We choose these
components as follows.

\subsubsection{The choice of transition matrices.}
We have three setups for choosing $P(0)$. The first setup is a
family of $4$ types of structured stochastic monotone matrices, which we
denote by ${\cal P}_{\ell}(p)$, $\ell \in \{1, \ldots, 4\}$, where $p \in [0,
1]$ is a parameter of the model. The second setup is a randomly generated
stochastic monotone matrices which we denote by ${\cal R}(d)$, where $d \in
[0, 1]$ is a parameter of the model. In the third setup, we generate random stochastic
matrices using Levy distribution. The details of
these models are presented in the supplementary material.

\subsubsection{The choice of the cost function.}
For all our experiments we choose $c(x, 0) = (x-1)^2$ and $c(x, 1) =
0.5(|\mathcal{X}|-1)^2$. 

\subsection{Experimental details and result}
We conduct different experiments to compare the performance of Whittle index
with the optimal policy and the myopic policy for different setups
(described in Section~\ref{sebsec:setup}) and for different sizes 
$|\mathcal{X}|$ of the state space, the number $n$ of the arms, and the number $m$ of
active arms. For all experiments we choose the discount factor $\beta = 0.95$. 

We evaluate the performance of a policy via Monte Carlo simulations over $S$
trajectories, where each trajectory is of length $T$.
In all our experiments, we choose $S
= 2500$ and $T = 250$.

\begin{figure}[!t]
	\centering
	\begin{subfigure}[b]{0.225\textwidth}   
		\centering 
		\includegraphics[page=1]{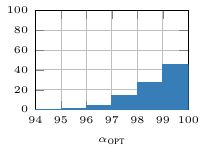}
		\caption[]%
		{{$m = 1$.}}    
		\label{fig:301}
	\end{subfigure}
	\qquad
	\begin{subfigure}[b]{0.225\textwidth}   
		\centering 
		\includegraphics[page=2]{histograms.pdf}
		\caption[]%
		{{$m = 2$.}}    
		\label{fig:302}
	\end{subfigure}
	\caption[]
	{Relative performance $\alpha_{\textsc{opt}}$
		of \textsc{wip} versus \textsc{opt} for Experiment~$2$.} 
	\label{fig:30}
\end{figure}
\begin{figure}[!t]
	\centering
	\begin{subfigure}[b]{0.22\textwidth}
		\centering
		\includegraphics[page=1,width=0.95\linewidth]{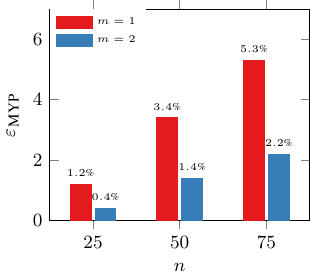}
		\caption[]%
		{{$\ell = 1$}}    
		\label{fig:25}
	\end{subfigure}
	\quad
	\begin{subfigure}[b]{0.225\textwidth}  
		\centering 
		\includegraphics[page=2,width=0.95\linewidth]{bars.pdf}
		\caption[]%
		{{$\ell = 2$}}    
		\label{fig:26}
	\end{subfigure}
	\quad
	\begin{subfigure}[b]{0.22\textwidth}   
		\centering 
		\includegraphics[page=3,width=0.95\linewidth]{bars.pdf}
		\caption[]%
		{{$\ell = 3$}}    
		\label{fig:27}
	\end{subfigure}
	\quad
	\begin{subfigure}[b]{0.225\textwidth}   
		\centering 
		\includegraphics[page=4,width=0.95\linewidth]{bars.pdf}
		\caption[]%
		{{$\ell = 4$}}    
		\label{fig:28}
	\end{subfigure}
	\caption[]
	{Relative improvement $\varepsilon_{\textsc{myp}}$ of
		\textsc{wip} vs.\ \textsc{myp}
		for Experiment~$3$.} 
	\label{fig:set5-8}
\end{figure}

\subsubsection*{\textbf{Experiment 1)} Comparison of Whittle index with the
optimal policy for structured models.}
The optimal policy is computed by solving the MDP for Problem~\ref{prob:main}, which is feasible only for small values of $|\mathcal{X}|$ and $n$.
We choose $|\mathcal{X}| = 5$ and $n = 5$
and compare the two policies for model ${\cal P}_{\ell}(\cdot)$, $\ell \in
\{1, \ldots, 4\}$ and $m \in \{1, 2\}$. 

For a given value of $n$ and $\ell$, we pick $n$ equispaced points $(p_1, \dots, p_n)$ in the interval $[0.35, 1]$ and choose $\mathcal{P}_{\ell}(p_i)$ as the transition matrix of arm~$i$. We observed that 
\(
\alpha_{\textsc{opt}} = J(\textsc{opt})/J(\textsc{wip})
\),
the relative (percentage) performance improvement of \textsc{wip} compared to \textsc{opt}, was in the range of 99.95\%--100\% for all parameters.

\subsubsection*{\textbf{Experiment 2)} Comparison of Whittle index with the
optimal policy for randomly sampled models.}

As before, we pick $|\mathcal{X}| = 5$ and $n = 5$ so that it is feasible to
calculate the optimal policy. For each arm, we sample the transition matrix
from ${\cal R}(5/|\mathcal{X}|)$ and repeat the experiment $250$ times. The histogram of $\alpha_{\textsc{opt}}$ over experiments for $m \in \{1, 2\}$ is shown in Fig~\ref{fig:30}, which show that \textsc{wip} performs close to \textsc{opt} in all cases.

\subsubsection*{\textbf{Experiment 3)} Comparison of Whittle index with the
myopic policy for structured models.} 
We generate the structured models as in Experiment~$1$ but for $|\mathcal{X}| = 25$,
$n \in \{25, 50, 75\}$, and $m \in \{1, 2, 5\}$. In this case, let
\(
\varepsilon_{\textsc{myp}} = (
J(\text{\textsc{myp}})-J(\text{\textsc{wip}}))/J(\text{\textsc{myp}})
\)
denote the relative improvement of \textsc{wip} compared to \textsc{myp}. The
results of $\varepsilon_{\textsc{myp}}$ for different choice of the parameters
are shown in Fig~\ref{fig:set5-8}.

In Fig~\ref{fig:set5-8}, we observe that \textsc{wip} performs considerably
better than \textsc{myp}. In addition to that, performance of \textsc{wip} is
better with respect to \textsc{myp} when $\ell = 4$ which is more complicated
than models where $\ell \in \{1, 2, 3\}$. However, increasing $m$ doesn't
necessarily contribute to better $\varepsilon_{\textsc{myp}}$ as overlap
between the choices of the two policies may increase. Note that as ${\cal
	P}_4(\cdot)$ is very different from the rest of the models, the trend of bars
in Fig~\ref{fig:28} with respect to $n$ varies differently from the rest of
the models. 

\begin{figure}[!t]
  \begin{minipage}{0.475\textwidth}
    \centering
    \includegraphics[page=1, width=\linewidth]{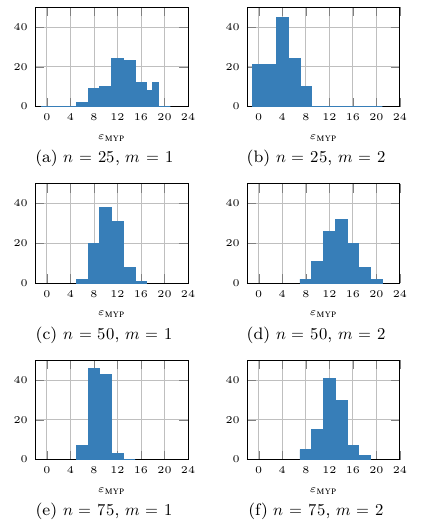}
	\caption
	{Relative improvement $\varepsilon_{\textsc{myp}}$ 
		of \textsc{wip} vs.\ \textsc{myp} 
    for Experiment~$4$.} 
	\label{fig:set9-12}
  \end{minipage}
  \hfill
  \begin{minipage}{0.475\textwidth}
    \centering
    \includegraphics[page=2, width=\linewidth]{improvement.pdf}
	\caption
	{Relative improvement $\varepsilon_{\textsc{myp}}$ 
		of \textsc{wip} vs.\ \textsc{myp}
    for Experiment~$5$.} 
	\label{fig:set12-15}
  \end{minipage}
\end{figure}

\subsubsection*{\textbf{Experiment 4)} Comparison of Whittle index with the myopic policy for randomly sampled models}

We generate $250$ random models as described in Experiment~$2$ but for $|\mathcal{X}|
= 25$ and larger values of $n$. For each case, $\varepsilon_{\textsc{myp}}$ is
computed. The histogram of $\varepsilon_{\textsc{myp}}$ for different choices
of the parameters are shown in Fig~\ref{fig:set9-12}.

The result shows that on average, \textsc{wip} performs considerably better than \textsc{myp} and this improvement is guaranteed as the concentration of data for the sampled models is mostly on positive values of $\varepsilon_{\textsc{myp}}$. 

\subsubsection*{\textbf{Experiment 5)} Comparison of Whittle index with the
myopic policy for restart models.}

We generate $250$ random stochastic matrices for $P(0)$.\footnote{Each row of the matrix is generate according to Section 1.3 of the supplementary material.} We set $|\mathcal{X}| = 25$ and $n \in \{25, 50, 75\}$ and $m \in \{1, 2\}$. For each case, $\varepsilon_{\textsc{myp}}$ is computed and the histogram of $\varepsilon_{\textsc{myp}}$ for different choices of the parameters is shown in Fig~\ref{fig:set12-15}.

\section{Conclusion} \label{sec:conclude}
We present two general sufficient conditions for restless bandit processes to
be indexable. The first condition depends only on the transition matrix $P(1)$
while the second condition depends on both $P(0)$ and $P(1)$. These sufficient
conditions are based on alternative characterizations of the passive set,
which might be useful in general as well.
We also present refinements of these sufficient conditions that are simpler to
verify. Two of these simpler conditions are worth highlighting: models where
the active action resets the state according to a known distribution and
models where the discount factor is less than $0.5$.

We then present a generalization of a previous proposed adaptive greedy algorithm, which was developed to compute the Whittle index for a sub-class of restless bandits known as PCL indexable bandits. We show that the generalized adaptive greedy algorithm computes the Whittle index for all indexable bandits. \textcolor{black}{We provide a computationally efficient implementation of our algorithm, which computes the Whittle indices of a restless bandit with $K$ states in $\OO(K^3)$ computations.}

Finally, we show how to refine the results for two classes for restless
bandits: stochastic monotone bandits and restless bandits with controlled
restarts. We also present a detailed numerical study which shows that Whittle
index policy performs close to the optimal policy and considerably better than
a myopic policy.

\appendix

\section{Proof of Proposition~\ref{prop:characterization}}
\label{app:characterization}

We first present a preliminary result.
\begin{lemma}\label{lemma:h0}
	For $\tau = 0$, the policy $h_0$ satisfies $J^{(h_0)}_\lambda(x) =
    H_\lambda(x, 1) = (1-\beta)c(x, 1)+W_\lambda$.
\end{lemma}
\begin{proof}
  Consider the stopping time~$\tau = 0$. The policy~$h_0,$ takes the active
  action at time~$0$ and follows the optimal policy afterwards. Thus, for any
  $x \in \mathcal X$, $J^{(h_0)}(x) = (1-\beta)(c(x,1) + \lambda) + \beta \sum_{y \in
  \mathcal{X}} P_{xy}(1) V_\lambda(y) =  H_\lambda(x, 1)$. By~\eqref{eqn:def_H}
  and~\eqref{eqn:W_def} we have $H_\lambda(x, 1) = (1-\beta)c(x, 1)+W_\lambda(x)$.
\end{proof}

We now proceed with the proof of Proposition~\ref{prop:characterization}.
By definition, $\Pi^{(a)}_\lambda = \Pi_\lambda$. We establish the equality of
other characterizations.

\begin{enumerate}
  \item[(i)] $\Pi^{(a)}_\lambda = \Pi^{(b)}_\lambda$.
    We have
    \(
      x \in \Pi_\lambda \overset{(a)}{\iff} g_\lambda(x) = 0 \overset{(b)}{\iff} H_\lambda(x, 0) < H_\lambda(x, 1)
    \)
    where $(a)$ follows from~\eqref{eqn:pass_set} and $(b)$ follows from the dynamic program~\eqref{eqn:bllmn_vf}. 

  \item[(ii)] $\Pi^{(b)}_\lambda \subseteq \Pi^{(c)}_\lambda$.
    Let $\sigma$ denote the hitting time of $\mathcal{X} \setminus
    \Pi_\lambda$. If we start in state~$x \in \Pi^{(b)}_\lambda =
    \Pi_\lambda$, then the policy~$h_{\sigma, \lambda}$ is same as the optimal policy.
    Hence, $J^{(h_{\sigma, \lambda})}_{\lambda}(x) = H_\lambda(x, 0)$. Thus, for any $x \in \Pi^{(b)}_\lambda = \Pi_\lambda$, 
    $J^{(h_{\sigma, \lambda})}_{\lambda}(x) = H_\lambda(x, 0) \overset{(a)}{<} H_\lambda(x,1) \overset{(b)}{=} J^{(h_0)}_{\lambda}(x)$
    where $(a)$ follows from fact that $x \in \Pi^{(b)}_\lambda$
    and $(b)$ from Lemma~\ref{lemma:h0}.

  \item[(iii)] $\Pi^{(c)}_\lambda \subseteq \Pi^{(b)}_\lambda$.
    Let $x \in \Pi^{(c)}_\lambda$ and $\sigma \in \Sigma$ denote a stopping time such that $J^{(h_{\sigma, \lambda})}_{\lambda}(x) < J^{(h_0)}_{\lambda}(x)$. Now, the optimal policy performs at least as well as policy $h_{\sigma, \lambda}$. Therefore, $V_\lambda(x) \leq J^{(h_{\sigma, \lambda})}_{\lambda}(x)$. Combining this result with Lemma~\ref{lemma:h0} we have $V_\lambda(x) < H_\lambda(x, 1)$. Thus, we must have $V_\lambda(x) = H_\lambda(x, 0)$ which results in $H_\lambda(x, 0) < H_\lambda(x, 1)$ which implies $x \in \Pi^{(b)}_\lambda$.

  \item[(iv)] $\Pi^{(c)}_\lambda = \Pi^{(d)}_\lambda$.
    According to the definitions of $L(x, \tau)$ and $W_\lambda(x)$ we have
    \begin{align}
      J^{(h_{\tau, \lambda})}_{\lambda}(x) = (1-\beta) L(x, \tau) + \EXP [ \beta^{\tau} W_\lambda(X_\tau) | X_0 = x ]. \label{eqn:J_htau}
    \end{align}
    Thus, $J^{(h_{\sigma, \lambda})}_{\lambda}(x) < J^{(h_0)}_{\lambda}(x)$ if and only if
    \begin{align}
      (1-\beta) L(x, \sigma) + \EXP [ \beta^{\sigma} W_\lambda(X_\sigma) | X_0 = x ] < (1 - \beta) c(x, 1) + W_\lambda(x) \label{eqn:n_cond}
    \end{align}
    where we have used~\eqref{eqn:J_htau} for $J^{(h_{\sigma, \lambda})}_{\lambda}(x)$ and
    Lemma~\ref{lemma:h0} for $J^{(h_0)}_{\lambda}(x)$. Rearranging the terms
    of~\eqref{eqn:n_cond} we get the expression in $\Pi^{(d)}_\lambda$. Hence,
    $\Pi^{(c)}_\lambda = \Pi^{(d)}_\lambda$.
\end{enumerate}

\section{Proof of Theorem~\ref{Thm:suf_1}}\label{prf:Thm:suf_1}

\subsection{Proof of Theorem~\ref{Thm:suf_1}.a}

We first present a preliminary result. Let $\Delta_\lambda := \lambda'' - \lambda'$ for any 
\begin{lemma}\label{lemma:W_inclam}
  Under~\eqref{eqn:P1_cond}, for any $\lambda'' > \lambda'$ and $\sigma \in
  \Sigma$, $\sigma \neq 0$, we have that for any $x \in \mathcal X$, 
	\[ 
      W_{\lambda'}(x) - \EXP [ \beta^{\sigma} W_{\lambda'}(X_\sigma) | X_0 = x ]
      \leq  
      W_{\lambda''}(x) - \EXP [ \beta^{\sigma} W_{\lambda''}(X_\sigma) | X_0 = x ], 
    \]
\end{lemma}
\begin{proof}
  By~\eqref{eqn:W_def}, we have for any $x \in \mathcal{X}$,
  \begin{multline}
    (W_{\lambda''}(x) - \EXP [ \beta^{\sigma} W_{\lambda''}(X_\sigma) | X_0 = x]) 
    -
    (W_{\lambda'}(x) - \EXP [ \beta^{\sigma} W_{\lambda'}(X_\sigma) | X_0 = x]) 
    \\
    = (1-\beta) \Delta_\lambda \bigl(1 - M(x, \sigma)\bigr) 
    + \beta \EXP\biggl[ \sum_{y \in \mathcal{X}} \bigl( P_{xy}(1) -
      \beta^\sigma P_{X_\sigma y}(1) \bigr) \bigl( V_{\lambda''}(y) -
    V_{\lambda'}(y) \bigr) \biggm| X_0 = x \biggr]
    \label{eq:split-1}
  \end{multline}
  Now since $\sigma \ge 1$, $M(x,\sigma) \le \beta$ and, 
  \begin{equation} \label{eq:split-1-1}
    (1-\beta) \Delta_\lambda(1 - M(x,\sigma)) \ge 
    \Delta_\lambda (1-\beta)^2
  \end{equation}
  Now consider,
  \begin{align}
    \hskip 1em & \hskip -1em
    \beta \EXP\biggl[ \sum_{y \in \mathcal{X}} \bigl( P_{xy}(1) -
      \beta^\sigma P_{X_\sigma y}(1) \bigr) \bigl( V_{\lambda''}(y) -
    V_{\lambda'}(y) \bigr) \biggm| X_0 = x \biggr] \notag \\
    & \stackrel{(a)}{\geq} \beta \EXP\biggl[
      \sum_{y \in \mathcal{X}} \bigl( P_{xy}(1) - \beta P_{X_\sigma y}(1) \bigr)
    \bigl( V_{\lambda''}(y) - V_{\lambda'}(y) \bigr) \biggm| X_0 = x \biggr]
    \notag \\
    & \stackrel{(b)}{\geq} \beta \Delta_\lambda \EXP \biggl[ 
      \sum_{y \in \mathcal{X}} \Bigl\{ 
        \bigl[ P_{xy}(1) - \beta P_{X_\sigma y}(1) \bigr]^{+} N^{(g_{\lambda''})}(y) 
        \notag \\
        & \hskip 6em
        + \bigl[ P_{xy}(1) - \beta P_{X_\sigma y}(1) \bigr]^{-}
    N^{(g_{\lambda'})}(y) \Bigr\} \bigg| X_0 = x \biggr] 
    \notag \\
    &\stackrel{(c)}\ge - \Delta_\lambda (1-\beta)^2,
    \label{eq:split-1-2}
  \end{align}
  where $(a)$ holds due to $\sigma \geq 1$ and $(b)$ holds by
  Lemma~\ref{lem:V-diff} and~$(c)$ follows from~\eqref{eqn:P1_cond}.
  Substituting~\eqref{eq:split-1-1} and~\eqref{eq:split-1-2}
  in~\eqref{eq:split-1}, we get the result of the Lemma.
\end{proof}

We now proceed with the proof of Theorem~\ref{Thm:suf_1}a. 
Consider $\lambda' < \lambda''$. Suppose $x \in \Pi_{\lambda'}$. By
Proposition~\ref{prop:characterization}.d, there exists a $\sigma \neq 0$ such that
$(1-\beta) \left( L(x, \sigma) - c(x, 1) \right) < W_{\lambda'}(x) - \EXP
[ \beta^{\sigma} W_{\lambda'}(X_\sigma) | X_0 = x ]. $ Combining this
result with the result of Lemma~\ref{lemma:W_inclam}, we infer 
$(1-\beta)
\left( L(x, \sigma) - c(x, 1) \right) < W_{\lambda''}(x) - \EXP [
\beta^{\sigma} W_{\lambda''}(X_\sigma) | X_0 = x ].$
Thus, $x \in \Pi_{\lambda''}$. Hence, $\Pi_{\lambda'}\subseteq \Pi_{\lambda''}$ and the
RB is indexable.

\subsection{Proof of Theorem~\ref{Thm:suf_1}.b} \label{prf:Thm:suf_2}
Consider $\lambda' < \lambda''$. A RB is indexable if
$\Pi_{\lambda'}\subseteq \Pi_{\lambda''}$ or equivalently, for any $x$ such
that $H_{\lambda'}(x, 0) < H_{\lambda'}(x, 1)$ then $H_{\lambda''}(x, 0) <
H_{\lambda''}(x, 1)$. A sufficient condition for that is to show that
$H_{\lambda'}(x, 1) - H_{\lambda'}(x, 0) \leq H_{\lambda''}(x, 1) -
H_{\lambda''}(x, 0)$, or equivalently, show that $H_{\lambda''}(x, 0) -
H_{\lambda'}(x, 0) \leq H_{\lambda''}(x, 1) - H_{\lambda'}(x, 1)$. We prove
this inequality as follows.

Let $\Delta_\lambda = \lambda'' - \lambda'$. By~\eqref{eqn:def_H}, we have for any $x \in \mathcal{X}$,
\begin{align*}
  \hskip 1em & \hskip -1em	
  \left(H_{\lambda''}(x, 1) - H_{\lambda'}(x, 1) \right) - \left(H_{\lambda''}(x, 0) - H_{\lambda'}(x, 0) \right) \notag \\ 
	& = \Delta_\lambda(1-\beta) + \beta \sum_{y \in \mathcal{X}} (P_{xy}(1) - P_{xy}(0)) (V_{\lambda''}(y) - V_{\lambda'}(y)) \notag \\
	& \stackrel{(a)}{\geq} \Delta_\lambda \Bigl( 1-\beta + \beta \sum_{y \in
    \mathcal{X}} \left[ P_{xy}(1) - P_{xy}(0) \right]^{+} N^{(g_{\lambda''})}(y)
  + \left[ P_{xy}(1) - P_{xy}(0) \right]^{-} N^{(g_{\lambda'})}(y) \Bigr)
	\stackrel{(b)}{\geq} 0 
\end{align*}
where $(a)$ follows from Lemma~\ref{lem:V-diff} and $(b)$ holds
by~\eqref{eqn:P2_cond}. Therefore the RB is indexable.

\section{Proof of Proposition~\ref{prop:1}}
\label{prf:prop:1}

We prove the result of each part separately. 
\begin{enumerate}
	\item[a.] This follows from observing that 
      \begin{align*}
        \hskip 1em & \hskip -1em	
        \sum_{y \in \mathcal{X}}
        \Bigl\{
          \bigl[ \beta P_{zy}(1) - P_{xy}(1) \bigr]^+ N^{(g)}(y) 
          -
          \bigl[ P_{xy}(1) - \beta P_{zy}(1) \bigr]^+ N^{(h)}(y) 
        \Bigr\}
        \\
        &\stackrel{(a)}\le 
        \sum_{y \in \mathcal{X}}
          \bigl[ \beta P_{zy}(1) - P_{xy}(1) \bigr]^+ N^{(g)}(y) 
        \\
        &\stackrel{(b)}\le 
        \sum_{y \in \mathcal{X}}
          \bigl[ \beta P_{zy}(1) - P_{xy}(1) \bigr]^+ 
        \le 
        \max_{x,z \in \mathcal{X}}
        \sum_{y \in \mathcal{X}}
          \bigl[ \beta P_{zy}(1) - P_{xy}(1) \bigr]^+ 
      \end{align*}
      where we are ignoring negative terms in $(a)$ and using 
      $N^{(g)}(x) \le 1$ in~$(b)$.

    \item[b.] For any $x,y,z \in \mathcal{X}$, 
      $P_{xy}(1) - \beta P_{zy}(1) = (1-\beta) P_{xy}(1)$. Thus,
      \begin{align*}
        \hskip 1em & \hskip -1em	
        \sum_{y \in \mathcal{X}}
        \Bigl\{
          \bigl[ \beta P_{zy}(1) - P_{xy}(1) \bigr]^+ N^{(g)}(y) 
          -
          \bigl[ P_{xy}(1) - \beta P_{zy}(1) \bigr]^+ N^{(h)}(y) 
        \Bigr\}
        \\
        &= -
        \sum_{y \in \mathcal{X}}
        (1-\beta) P_{xy}(1) N^{(h)}(y) 
        \le 0 < \frac{(1-\beta)^2}{\beta}.
      \end{align*}

    \item[c.] This follows from observing that
      \begin{align*}
        \hskip 1em & \hskip -1em	
        \sum_{y \in \mathcal{X}} 
        \Bigl\{
          \bigl[ P_{xy}(0) - P_{xy}(1) \bigr]^+ N^{(g)}(y) 
          -
          \bigl[ P_{xy}(1) - P_{xy}(0) \bigr]^+ N^{(h)}(y) 
        \Bigr\}
        \\
        &\stackrel{(a)}\le
        \sum_{y \in \mathcal{X}} 
          \bigl[ P_{xy}(0) - P_{xy}(1) \bigr]^+ N^{(g)}(y) 
        \\
        &\stackrel{(b)}\le
        \sum_{y \in \mathcal{X}} 
          \bigl[ P_{xy}(0) - P_{xy}(1) \bigr]^+ 
        \le
        \max_{x\in \mathcal{X}}
        \sum_{y \in \mathcal{X}} 
          \bigl[ P_{xy}(0) - P_{xy}(1) \bigr]^+ 
      \end{align*}
      where we are ignoring negative terms in $(a)$ and using 
      $N^{(g)}(x) \le 1$ in~$(b)$.

    \item[d.] $\beta \le 0.5$ implies that
      \[
        \dfrac{1-\beta}{\beta} \ge 1 \ge \max_{x \in \mathcal X}
          \bigl[ P_{xy}(0) - P_{xy}(1) \bigr]^+ 
      \]
      which is the same as sufficient condition~(c) established above.
\end{enumerate}

\section{Proof of Lemma~\ref{lemma:WJ}}\label{prf:lemma:WJ}
\textcolor{black}{The proof of each part is as follows:
\begin{enumerate}
	\item Since the model is indexable and \textcolor{black}{$y \in \Gamma_{d+1}$}, $w(d) = \lambda_{d+1}$. Therefore, the optimal policy is indifferent between choosing the active and the passive action at $\lambda = \lambda_{d+1}$.
	\item By definition, for any $\lambda \in (\lambda_d, \lambda_{d+1}]$, $h_d$ is an optimal policy. Therefore, we have $J^{({h}_{d, y})}_\lambda(x) \geq J^{(h_d)}_\lambda(x)$ with $y \in {\cal X}\backslash{\cal P}_d$, for all $x \in {\cal X}$ with equality if \textcolor{black}{$y \in \Gamma_{d+1}$} and $\lambda = \lambda_{d+1}$. 
\end{enumerate}}

\bibliographystyle{APT}
\bibliography{mybibfile}

\end{document}